\documentclass[letterpaper]{article} 
\usepackage{aaai25}  
\usepackage{times}  
\usepackage{helvet}  
\usepackage{courier}  
\usepackage[hyphens]{url}  
\usepackage{graphicx} 
\usepackage{natbib}  
\usepackage{caption} 
\usepackage{algorithm}
\usepackage{algorithmic}
\usepackage{amsmath}
\usepackage{amsthm}
\usepackage{amssymb}
\usepackage{booktabs}
\usepackage{enumerate}

\newtheorem{definition}{Definition}
\newtheorem{lemma}{Lemma}
\newtheorem{theorem}{Theorem}
\newtheorem{proposition}{Proposition}
\newtheorem{corollary}{Corollary}
\newtheorem{claim}{Claim}

\usepackage{newfloat}
\usepackage{listings}
\DeclareCaptionStyle{ruled}{labelfont=normalfont,labelsep=colon,strut=off} 
\floatstyle{ruled}
\newfloat{listing}{tb}{lst}{}
\floatname{listing}{Listing}
\title{Merging Mechanisms for Ads and Organic Items in E-commerce Platforms}
\author{
    Nan An\textsuperscript{\rm 1}\equalcontrib,
    Weian Li\textsuperscript{\rm 2}\equalcontrib,
    Qi Qi\textsuperscript{\rm 1}\thanks{Corresponding Author.},
    Liang Zhang\textsuperscript{\rm 1}
}
\affiliations{
    \textsuperscript{\rm 1} Gaoling School of Artificial Intelligence, Renmin University of China, Beijing, China\\
    \textsuperscript{\rm 2} School of Software, Shandong University, Jinan, China\\
    \{annan0425, qi.qi, zhang.liang\}@ruc.edu.cn,
    weian.li@sdu.edu.cn  

}

\usepackage{bibentry}

\begin{document}

\maketitle

\begin{abstract}
In contemporary e-commerce platforms, search result pages display two types of items: ad items and organic items. Ad items are determined through an advertising auction system, while organic items are selected by a recommendation system. These systems have distinct optimization objectives, creating the challenge of effectively merging these two components. Recent research has explored merging mechanisms for e-commerce platforms, but none have simultaneously achieved all desirable properties: incentive compatibility, individual rationality, adaptability to multiple slots, integration of inseparable candidates, and avoidance of repeated exposure for ads and organic items. This paper addresses the design of a merging mechanism that satisfies all these properties. We first provide the necessary conditions for the optimal merging mechanisms. Next, we introduce two simple and effective mechanisms, termed the generalized fix mechanism and the generalized change mechanism. Finally, we theoretically prove that both mechanisms offer guaranteed approximation ratios compared to the optimal mechanism in both simplest and general settings.
\end{abstract}

%

\section{Introduction}
\label{sec:intro}

The rapid advancement of the Internet has established e-commerce platforms as leading online marketplaces, offering unparalleled convenience for both sellers and buyers. Sellers can swiftly set up online stores, while buyers can easily search for products using just a few keywords. However, this ease of use presents a significant challenge for platforms: how to effectively and accurately present products. Consequently, e-commerce platforms are continuously working to enhance their methods for generating relevant and suitable product listings.

Current search result pages typically feature two main components: organic items and advertising items. Organic items are selected by a \emph{recommendation system}, which prioritizes user experience and purchase intent to boost overall sales. In contrast, advertising items are determined through a \emph{sponsored search auction}, where advertisers bid for display slots and pay extra fees for increased exposure. This auction system has become a major revenue source for online platforms, generating an impressive \$225 billion in 2023. The challenge lies in integrating these two types of outputs into a cohesive list, a complex multi-objective optimization problem that has attracted significant attention from both academia and industry.

Several studies have explored this problem using various approaches. Traditionally, platforms allocate the top slots in search results to advertising items, with the recommendation system and sponsored search auction operating independently. This separation makes it difficult to achieve a globally optimal arrangement. Recently, some companies have adopted merging methods, incorporating winning advertising items into the recommendation system's candidate set and using machine learning to produce a unified list. Although this approach has improved average performance, it lacks theoretical support and does not always preserve desirable auction theory properties. Theoretical studies have also aimed to design mechanisms that address both organic items and advertisements, but they often struggle with issues such as repeated exposure (where the same product appears as both an ad and an organic item) and adaptability to multi-slot scenarios (which are prevalent in search engines or e-commerce platforms where several items are displayed on a single page).

Based on existing research, we observe that current methods fall short in addressing several key questions simultaneously: 1. How can we design a mechanism that effectively incorporates both ads and organic items while optimizing multiple objectives? 2. How can we maintain essential properties of advertising auctions, such as incentive compatibility (IC) and individual rationality (IR)? 3. How can we avoid the repeated display of an item in both its ad and organic forms? 4. How can we adapt to a multi-slot setting effectively? Each of these issues is both theoretically significant and practically valuable.

\subsection{Main Contribution}



To address the aforementioned questions, this work tackles the mechanism design problem for e-commerce platforms in a multi-slot setting. In this scenario, several candidate items can be displayed in two forms: as ads or as organic content. Our focus is on developing IC and IR merging mechanisms that effectively balance revenue and user experience, while avoiding the repeated display of the same item in both its ad and organic forms.

Our work extends the problem studied by \cite{LMZWZYXZD24}, which addresses a similar issue in the single-slot context. We generalize their model to the more prevalent multi-slot scenario and propose two desirable mechanisms with guaranteed approximate ratios. Our main results are summarized as follows:
\begin{itemize}
    \item In addition to the known characterization of feasible mechanisms, we introduce additional necessary conditions for the optimal mechanisms in both general and the simplest settings (three items and two slots). Notably, in the simplest setting, we find that a partial order always exists for displaying organic items. These conditions facilitate the construction of our mechanisms and simplify the subsequent performance analysis. We propose two simple mechanisms: the generalized fix (G-FIX) mechanism and the generalized change (G-CHANGE) mechanism, both of which are IC, IR, and applicable to the multi-slot setting.
    \item \textbf{G-FIX mechanism}: In this mechanism, we restrict each item to participate in the mechanism in only one form, either as an ad or as an organic item, thereby avoiding overlap and simplifying the allocation process. We then select the most suitable items for display. The G-FIX mechanism is shown to be a $(4/5)^3$-approximate mechanism in the simplest setting. For a general setting with $n$ items and $k$ slots, under certain assumptions, the G-FIX mechanism is nearly optimal.
    \item \textbf{G-CHANGE mechanism}: Unlike the G-FIX mechanism, the G-CHANGE mechanism allows a set of ordered items to participate in both forms (ad and organic). Items are evaluated in reverse order to determine their displayed form, which mitigates the effect of the other items and ensures compliance with the IC condition. In the simplest setting, we demonstrate that the G-CHANGE mechanism is optimal. In a general setting, the G-CHANGE mechanism is a $1/2$-approximate mechanism compared to the optimal merging mechanisms, without any assumption.
\end{itemize}
Due to space limitation, complete proofs can be referred to the full version. 

\subsection{Related Works}

Our contribution lies within the field of Internet economics, specifically in the design of sponsored search auctions. Our work is related to recent theoretical paper \cite{LMZWZYXZD24}, which considers a mechanism design problem to avoid the repeated display of ads and organic items, but their models are limited to the single-slot setting. In comparison to this paper, we generalize their models to multi-slot setting and design desirable mechanisms with guaranteed approximate ratios.  

In the realm of online advertising and recommendation systems, significant research has been devoted to integrating ad and organic items. Our work addresses a multi-objective auction design problem that incorporates the impact of organic items. Here, we briefly provide the related studies in online advertising.
Initially, sponsored advertising was implemented through generalized first-price (GFP) and generalized second-price (GSP) auctions. Subsequent research explore equilibrium analysis \cite{EOS07,V07,LT10,LH15} and revenue enhancement \cite{LP07,TL13} for GSP and GFP auctions. Notably, \citet{BCKKK14} investigate multi-objective optimization in the context of sponsored advertising.

Regarding the integration of ads and organic items, two primary directions have emerged. The first direction is to utilize advanced machine learning techniques to combine ads and organic items into a unified list. For example, several studies \cite{WJHCY19,ZZYLT20,LWW22,ZXLLTZ23} model this process as a Markov decision process and apply reinforcement learning to generate the item list. Additionally, \citet{LYZ21} and \citet{LWZ24} employ automated mechanism design to integrate one ad with multiple organic items. The second direction focuses on optimization methods, particularly Lagrangian duality. Papers in this category \cite{ZWM18,YXT20,CWN23} typically create a virtual ranking score using Lagrangian duality and allocate slots based on this score. Particularly, \citet{LQWY23} addresses the multi-slot setting and proposes a Myerson-like merge mechanism aimed at achieving the optimal objective values. However, their mechanism does not effectively address the issue of repeated exposure of items in both ad and organic forms.

Observing from the aforementioned papers, we can find that merging ads and organic items is a popular research topic in recent years. Although there have been fruitful results in both academia and industry, these works have not addressed the four questions mentioned earlier in this section simultaneously, and at least one question remains unanswered. To the best of our knowledge, our paper is the first work to design IC and IR mechanisms that directly handle and avoid the overlap of ads and organic items, while also being applicable to the multi-slot model.

\section{Model and Preliminaries}
\label{sec:pre}

\subsection{Merging Mechanism}
In this section, we formally establish our model for developing a merging mechanism. Consider a search results page with $k$ slots for displaying either ads or organic items. Suppose that there are $n$ candidate items, denoted by $i \in[n]$. 
Each item $i$ can be displayed in one of two forms: as an ad or as an organic item. We denote these forms by $\text{Ad}_{i}$ for the ad form and $\text{Org}_{i}$ for the organic form. We assume that each item can be displayed at most once and in only one form, either as $\text{Ad}_{i}$ or $\text{Org}_{i}$. Therefore, $\text{Ad}_{i}$ and $\text{Org}_{i}$ will not appear simultaneously on the same search results page.


The two forms of an item, $\text{Ad}_{i}$ and $\text{Org}_{i}$, can be distinguished by the user and may have different effects. Let $\alpha_{i}^{A}$ represent the click-through rate (CTR) of the ad form $\text{Ad}_{i}$, and $\alpha_{i}^{O}$ represent the CTR of the organic form $\text{Org}_{i}$. We assume that $\alpha_{i}^{A}<\alpha_{i}^{O}$ due to users' preference for organic items. Additionally, we suppose that the CTRs for all items are public information.


Additionally, the different forms may affect user experience in varying ways. For each item $i$, let $\gamma_{i}^{A}$ and $\gamma_{i}^{O}$ denote the user experience effect of the ad form and organic form, respectively. Suppose that $0 \leq \gamma_{i}^{A} \leq \gamma_{i}^{O}$, since users prefer more on organic items.


With the above symbols in mind, we introduce the merging mechanisms used in this paper. Before the mechanism is implemented, each item $i \in [n]$\ (referred to as owner $i$ for convenience) submits a bid $b_i$\footnote{For simplicity, we use $b_i$ to represent both the bid and the item's value, as our focus is on incentive-compatible mechanisms.}. Let $\boldsymbol{b}=$ $(b_{1}, \cdots, b_{n})$ denote the bid profile. We use $\boldsymbol{b}_{-i}$ to refer to the bid profile excluding owner $i$. We assume that each bid $b_i$ is independently drawn from a prior distribution $F_i$ known to the platform. Each $F_i$ is further assumed to be a regular distribution\footnote{A regular distribution, as defined by \cite{M81}, satisfies that $b_i-(1-F_i(b_i))/f_i(b_i)$ is monotone non-decreasing, where $f_i$ is the p.d.f of $F_i$.} with support $[0, B_{i}]$ for some $B_{i} \geq 0$, which is a common assumption in economics. We now define a merging mechanism as follows.

\begin{definition}[Merging Mechanism]\label{def:merge mechanism}
Collecting bids from owners, a merging mechanism is denoted by $\mathcal{M}=(x, y, p)$, where:
\begin{enumerate}
\item $x(\boldsymbol{b})=(x_{i}(\boldsymbol{b}))_{i \in[n]}$ and $y(\boldsymbol{b})=(y_{i}(\boldsymbol{b}))_{i \in[n]}$ are the allocation rules for ad items and organic items, respectively. 
For each item $i \in [n], x_{i}(\boldsymbol{b})=1$  indicates that its ad form $\mathrm{Ad}_{i}$ is displayed, otherwise $x_{i}(\boldsymbol{b})=0$. Similarly, $y_{i}(\boldsymbol{b})=1$ means that its organic form  $\mathrm{Org}_{i}$ is displayed, $y_{i}(\boldsymbol{b})=0$ otherwise.
\item $p(\boldsymbol{b})=(p_{i}(\boldsymbol{b}))_{i \in[n]}$ represents the payments charged to the owners when their ads are clicked. Specifically, if $\mathrm{Ad}_{i}$ is displayed ($x_{i}(\boldsymbol{b})=1$), the platform charges a payment $p_{i}(\boldsymbol{b})$ if the user clicks on $\mathrm{Ad}_{i}$ (with probability $\alpha_{i}^{A}$). We assume that $p_{i}(\boldsymbol{b})=0$ if $x_{i}(\boldsymbol{b})=0$.
\end{enumerate}
\end{definition}

Given the definition of the merging mechanism, we define the expected click-through rate $X_{i}(\boldsymbol{b})$ and the expected payment $P_{i}(\boldsymbol{b})$ for item $i$ as follows:
\begin{equation*}
    X_{i}(\boldsymbol{b})=x_{i}(\boldsymbol{b}) \alpha_{i}^{A}
                         +y_{i}(\boldsymbol{b}) \alpha_{i}^{O}
\end{equation*}
and
\begin{equation*}
    P_{i}(\boldsymbol{b})=p_{i}(\boldsymbol{b}) \cdot x_{i}(\boldsymbol{b}) \alpha_{i}^{A} .
\end{equation*}
Here, $X_{i}(\boldsymbol{b})$ represents the total probability that item $i$ is clicked in either form, and $P_{i}(\boldsymbol{b})$ is the expected payment when the item’s ad form is displayed and clicked. The utility $U_i(\boldsymbol{b})$ of owner $i$ under the bid profile $\boldsymbol{b}$ is defined as:
\[ 
    U_i(\boldsymbol{b})=b_i X_{i}(\boldsymbol{b})- P_{i}(\boldsymbol{b}).
\]
In this paper, we assume that item owners are utility maximizers and do not differentiate between the displayed forms of their items.

Since at most $k$ items can be displayed and each item can appear in only one form (either ad or organic), the merging mechanism must satisfy the following feasibility constraints: 
\[
    \sum_{i \in [n]}\left(x_{i}(\boldsymbol{b})+y_{i}(\boldsymbol{b})\right) \leq k.
\]
And for all item $i \in[n]$,
\[
    x_{i}(\boldsymbol{b})+y_{i}(\boldsymbol{b}) \leq 1.
\]

We aim for the merging mechanism to satisfy two desirable properties: incentive compatibility (IC) and individual rationality (IR). 
\begin{definition}[Incentive Compatibility]
    A merging mechanism is incentive compatible if owners are motivated to submit their true bids. Formally, for any bid profile $\boldsymbol{b}$ and any owner $i \in [n]$, for any misreported bid $b'_{i}$, the following condition must hold:
    \[
    b_{i} X_{i}(b_{i}, \boldsymbol{b}_{-i})-P_{i}(b_{i}, \boldsymbol{b}_{-i}) \geq b_{i} X_{i}(b_{i}^{\prime}, \boldsymbol{b}_{-i})-P_{i}(b_{i}^{\prime}, \boldsymbol{b}_{-i}).
    \]
\end{definition}
This ensures that the utility of reporting $b_{i}$ truthfully is not less than the utility of reporting any other bid $b'_{i}$.

\begin{definition}[Individual Rationality]
    A merging mechanism is individually rational if each owner’s utility is always non-negative, meaning that payments never exceed bids. Formally, for any bid profile $\boldsymbol{b}$ and any owner $i \in [n]$, the following must hold:
    \[
        p_{i}(\boldsymbol{b}) \leq b_{i} .
    \] 
\end{definition}
This ensures that the payment $p_{i}(\boldsymbol{b})$ charged to the owner is less than or equal to their bid $b_{i}$, guaranteeing that their utility is never negative.

For an e-commerce platform, the objective includes two main components: revenue and user experience. The revenue target is defined as the expected total payment. Formally, it is given by:
\begin{equation*}
    \text{Rev}(\mathcal{M}):=\underset{\boldsymbol{b} \sim F}{\mathbb{E}}\left[\sum_{i \in[n]} p_{i}(\boldsymbol{b}) x_{i}(\boldsymbol{b}) \alpha_{i}^{A} \right],
\end{equation*} 
where $F$ is the joint distribution on $\boldsymbol{b}$. The user experience target is defined as the total expected user experience, i.e.,
\begin{equation*}
    \text{UE}(\mathcal{M}):=\underset{\boldsymbol{b}\sim F}{\mathbb{E}}\left[
            \sum_{i \in[n]}\left(x_{i}(\boldsymbol{b}) \gamma_{i}^{A}+y_{i}(\boldsymbol{b}) \gamma_{i}^{O}\right)
            \right].
\end{equation*}

Note that the user experience of an item is independent of whether it is clicked. Once an item is displayed, its user experience is realized. Therefore, the platform's objective combines both revenue and user experience:
\begin{equation*}
    \text{OBJ}(\mathcal{M}):=\text{Rev}(\mathcal{M})+\text{UE}(\mathcal{M}). 
\end{equation*}

This formulation can be generalized to any weighted combination of revenue and user experience by incorporating weight coefficients into $\gamma_{i}^{A}$ and $\gamma_{i}^{O}$. Thus, the merging mechanism design problem can be expressed as:
\begin{equation}\label{fml:optmization}
    \begin{array}{rl} 
    \max\limits_{\mathcal{M}} & \text{OBJ}(\mathcal{M}) \\
    \text {s.t.} & \text { IC, IR, Feasibility.} 
    \end{array}
\end{equation}

Let $\mathcal{M}^{*}$ denote the optimal mechanism for the problem. We say that a merging mechanism $\mathcal{M}$ is $\tau$-approximate if the ratio of $\text{OBJ}(\mathcal{M})$ to $\text{OBJ}(\mathcal{M}^{*})$ is at least $\tau$. Formally, $\mathcal{M}$ is $\tau$-approximate if $\text{OBJ}(\mathcal{M})/ \text{OBJ}(\mathcal{M}^{*})\geq \tau.$


\subsection{Characterization of IC, IR, and Feasibility}
In the rest of this section, we summarize key results from previous work that are crucial for characterizing IC, IR, and feasible merging mechanisms. These results will also aid in the construction and analysis of mechanisms in the subsequent sections. Notably, Myerson's seminal work provides a fundamental framework for understanding IC and IR mechanisms.
\begin{lemma}[\citet{M81}]\label{lem:myerson}
     A merging mechanism $\mathcal{M}$ satisfies IC and IR if and only if the following conditions are met:
    \begin{enumerate}
    \item \textbf{Allocation Monotonicity}: For any item \(i \in [n]\), given the bids \(\boldsymbol{b}_{-i}\) of the other items, its expected click-through rate is non-decreasing in \(b_i\). Formally, for any \(b_i\) and \(b_i'\) with \(b_i < b_i'\), it holds that
    \[
    X_i(b_i, \boldsymbol{b}_{-i}) \leq X_i(b_i', \boldsymbol{b}_{-i}).
    \]

    \item \textbf{Payment Identity}: For any item \(i \in [n]\) and any bid profile \(\boldsymbol{b}\), the payment \(P_i(\boldsymbol{b})\) is given by
    \[
    P_i(\boldsymbol{b}) = b_i  X_i(b_i, \boldsymbol{b}_{-i}) - \int_{0}^{b_i} X_i(t, \boldsymbol{b}_{-i}) \, dt.
    \]
    \end{enumerate}
\end{lemma}

        

Recent work by \citet{LMZWZYXZD24} addresses a similar problem in the single-slot setting and provides necessary conditions for an IC, IR, and feasible merging mechanism. These conditions also apply to the multi-slot case. One important condition is form stability, which states that the allocation of an organic item depends solely on the bids of other owners and is unaffected by its own bid.
\begin{lemma}[Form Stability \cite{LMZWZYXZD24}]\label{lem:original form stability}
    Suppose $\mathcal{M}$ satisfies IC, IR, and feasibility. Then, for any item  $i \in[n]$, the allocation $y_{i}(\boldsymbol{b})$ of an organic item depends only on the bids $\boldsymbol{b}_{-i}$ of other owners and is unaffected by the bid $b_i$ of the owner $i$. Formally, for any $\boldsymbol{b}_{-i}$ and for any $b_{i} \neq b_{i}^{\prime}$, we have:
    \[
    y_{i}\left(b_{i}, \boldsymbol{b}_{-i}\right)=y_{i}\left(b_{i}^{\prime}, \boldsymbol{b}_{-i}\right).
    \]
\end{lemma}

Leveraging Myerson's lemma (Lemma \ref{lem:myerson}) and the condition of form stability (Lemma \ref{lem:original form stability}), we can express the payment $p_i(\boldsymbol{b})$ for each item concisely and simplify the total payment into an elegant form.
\begin{lemma}[\citet{LMZWZYXZD24}]\label{lem:payment}
    Suppose that $\mathcal{M}$ satisfies IC, IR, and feasibility, and consequently form stability. For any bid profile $\boldsymbol{b}$, if $\mathrm{Ad}_i$ is displayed, then the payment $p_i(\boldsymbol{b})$ is given by
    \[
      p_i(\boldsymbol{b}) = \inf \{ b'_i \geq 0 : x_i(b'_i, \boldsymbol{b_{-i}}) = 1 \},
    \]
    which is the critical bid required for $\mathrm{Ad}_i$ to be displayed.
\end{lemma}

\begin{lemma}[\citet{LMZWZYXZD24}]\label{lem:revenue}
    If $\mathcal{M}$ satisfies IC, IR, feasibility and form stability, the revenue can be expressed as
    \[
    \mathrm{Rev}(\mathcal{M}) = \underset{\boldsymbol{b}\sim F}{\mathbb{E}} \left[ \sum_{i \in [n]} \phi_i(b_i)  x_i(\boldsymbol{b}) \alpha_i^A  \right],
    \]
    where $\phi_i(b_i)=b_i-(1-F_i(b_i))/f_i(b_i)$ represents Myerson's virtual value function.
\end{lemma}
Note that, although the conclusions in Lemma \ref{lem:payment} and \ref{lem:revenue} resemble the original results from  \citet{M81}, they cannot be directly derived from Myerson's lemma due to the presence of two possible displayed forms for each item in our model. Therefore, form stability (Lemma \ref{lem:original form stability}) is essential to establish these results.

Given the above lemmas, a merging mechanism $\mathcal{M}$ is fully characterized by the allocation rules $x(\boldsymbol{b})$ and $y(\boldsymbol{b})$ and can be denoted by $\mathcal{M}^{x,y} = (x, y)$. An ad item contributes to both revenue and user experience. Specifically, we define $a_i = \phi_i(b_i) \alpha_i^A + \gamma_i^A$ as the contribution of $\text{Ad}_i$, and $o_i = \gamma_i^O$ as the contribution of $\text{Org}_i$ to the objective function. By using \(\boldsymbol{a} = (a_1, \cdots, a_n)\) as a proxy for \(\boldsymbol{b}\), the optimization problem (\ref{fml:optmization}) can be reformulated as:
\begin{equation}\label{fml:optimization 2}
    \begin{array}{rl}
    \max\limits_{x,y} & \operatorname{OBJ}(\mathcal{M}^{x,y})=\underset{\boldsymbol{a} \sim G}{\mathbb{E}} \left[ \sum\limits_{i \in [n]} \left(a_i x_i(\boldsymbol{a}) + o_i y_i(\boldsymbol{a})\right)\right]  \\
    \text {s.t.} & \sum\limits_{i \in [n]} \left(x_i(\boldsymbol{a}) + y_i(\boldsymbol{a})\right) \leq k, \ \    \forall \boldsymbol{a},  \\
     & x_i(\boldsymbol{a}) + y_i(\boldsymbol{a})\leq 1,     \ \ \ \ \ \ \ \ \ \, \  \ \ \forall i \in [n],\boldsymbol{a},  \\
     &y_i(a_i, \boldsymbol{a}_{-i}) = y_i(a_i', \boldsymbol{a}_{-i}), \ \  \forall i \in [n], \boldsymbol{a}_{-i}, a_i, a_i', \\
     & x_i(a_i, \boldsymbol{a}_{-i}) \leq x_i(a_i', \boldsymbol{a}_{-i}),  \, \   \forall i \in [n], \boldsymbol{a}_{-i}, a_i < a_i', 
    \end{array}
\end{equation}
where $G_i$ is the (regular) distribution of \(a_i\), induced by $F_i$, and \(G = \times_{i \in [n]} G_i\) is the joint distribution on $\boldsymbol{a}$. For convenience, we will refer to this optimization problem as the merging problem.

Finally, we introduce a mathematical notation that will be frequently used in the following sections. Given a set $S$ of finite real numbers and $k\leq |S|$, we define $\text{max}^{(k)}$ as an operator that computes the sum of the largest $k$ numbers in $S$. For example, if $S=\{1,2,3\}$, then $\text{max}^{(2)} S = 3+2 =5$.

\section{Form Stability in the Optimal Mechanisms}
\label{sec:fs}
Before introducing our merging mechanisms, we first establish necessary conditions, specifically form stability, that the optimal mechanisms must satisfy for the merging problem. These conditions guide the construction of our mechanisms and facilitate performance analysis. In this section, we examine form stability in two scenarios: the general case with \(n\) items and \(k\) slots, and the simplest multi-slot case with 3 items and 2 slots. To avoid ambiguity, we refer to the mechanism in the former scenario as the \(k\)-out-of-\(n\) mechanism, and in the latter as the 2-out-of-3 mechanism. Additionally, for any owner \(i \in [n]\), we denote the allocation rules in the optimal mechanism \(\mathcal{M}^*\) by \(x^*_i(\boldsymbol{a})\), \(y^*_i(\boldsymbol{a})\), and \(X^*_i(\boldsymbol{a})\).

We begin by characterizing the optimal \(k\)-out-of-\(n\) mechanism. According to form stability of any merging mechanism (Lemma \ref{lem:original form stability}), an owner's change in bid does not affect the allocation of her organic item. This implies that if an item is displayed in its organic form in the optimal mechanism, it will continue to be displayed in that form regardless of changes in its own bid. Intuitively, knowing that an item must be displayed in its organic form reduces the problem to designing a \((k-1)\)-out-of-\((n-1)\) mechanism. Therefore, by the optimality, the displayed forms of the remaining \(n-1\) items in the optimal mechanism will also remain unchanged. Similarly, if an item is displayed in its ad form, as long as its displayed form remains unchanged, bid changes will not affect the others. We formally summarize this discussion in Proposition \ref{prop: k from n optimal}. 

\begin{proposition}\label{prop: k from n optimal}
    In the optimal $k$-out-of-$n$ mechanism satisfying IC and IR, for any $i \in [n]$, $\boldsymbol{a}_{-i}$, and $a_i$, the following three statements hold:

\begin{enumerate}
    \item If $y_i^*(a_i, \boldsymbol{a}_{-i}) = 1$, it holds that for any $j$ and any $a_i'$,  
    \begin{small}
    \begin{align*}
        y_j^*(a_i', \boldsymbol{a}_{-i}) = y_j^*(a_i, \boldsymbol{a}_{-i}) ~ \text{and} ~ x_j^*(a_i', \boldsymbol{a}_{-i}) = x_j^*(a_i, \boldsymbol{a}_{-i}).
    \end{align*}
    \end{small}
    \item If $x_i^*(a_i, \boldsymbol{a}_{-i}) = 0$, it holds that for any $j$ and $a_i' \leq  a_i$,
    \begin{small}
    \begin{align*}
        y_j^*(a_i', \boldsymbol{a}_{-i}) = y_j^*(a_i, \boldsymbol{a}_{-i}) ~ \text{and} ~
        x_j^*(a_i', \boldsymbol{a}_{-i}) = x_j^*(a_i, \boldsymbol{a}_{-i}).
    \end{align*}
    \end{small}
    \item If $x_i^*(a_i, \boldsymbol{a}_{-i}) = 1$, it holds that for any $j$ and  $a_i' > a_i$,
    \begin{small}
    \begin{align*}
        y_j^*(a_i', \boldsymbol{a}_{-i}) = y_j^*(a_i, \boldsymbol{a}_{-i}) ~ \text{and} ~
        x_j^*(a_i', \boldsymbol{a}_{-i}) = x_j^*(a_i, \boldsymbol{a}_{-i}).
    \end{align*}
    \end{small}
\end{enumerate}

\end{proposition}

As a corollary, in the optimal mechanism, if an owner's bid adjustment does not alter her own displayed form, it will not affect the displayed forms of others.

\begin{corollary}\label{cor:optimal form stability}
    In the optimal $k$-out-of-$n$ mechanism satisfying IC and IR, given any $\boldsymbol{a}_{-i}$, if owner $i$ changes from $a_i$ to $a_i'$ such that $X_i^*(a_i, \boldsymbol{a}_{-i}) = X_i^*(a_i', \boldsymbol{a}_{-i})$, then for any $j$, $X^*_j$ remains unchanged.
\end{corollary} 


From the characterization of the optimal \(k\)-out-of-\(n\) mechanism, it is clear that form stability holds when only one owner changes her bid. However, the outcome becomes uncertain when more than two owners change their bids simultaneously. In the rest of this section, we examine the simplest multi-slot scenario with three items and two slots, further exploring form stability when all owners' bids can change simultaneously. Generally, maintaining a fixed allocation of ad items, \(x(\boldsymbol{a})\), is challenging, so we focus on the allocation of organic items, \(y(\boldsymbol{a})\). Intuitively, if under a given bid profile, the optimal mechanism favors owner \(i\) over owner \(j\) (i.e., \(y_i(\boldsymbol{a}) = 1\) and \(y_j(\boldsymbol{a}) = 0\)) in terms of organic items, it indicates that owner \(i\) contributes more to the objective than owner \(j\), regardless of the bids. The following lemma supports this intuition by demonstrating that in the simplest setting, if there exists a bid profile where the optimal mechanism displays owner \(i\)'s organic form but not owner \(j\)'s, then it is impossible for the optimal mechanism to display owner \(j\)'s organic form without also displaying owner \(i\)'s in any bid profile.

\begin{lemma}\label{lem:2 from 3 fs}
    In the optimal $2$-out-of-$3$ mechanism, 
    if there exists some $\boldsymbol{a}$ such that $y_i^*(\boldsymbol{a})=1$ and $y_j^*(\boldsymbol{a})=0$, then for any $\boldsymbol{a}'$, it is not possible that $y_i^*(\boldsymbol{a}')=0$ but $y_j^*(\boldsymbol{a}')=1$.
\end{lemma}

Using Lemma \ref{lem:2 from 3 fs}, we can compare the priority of any two items' organic forms in the optimal mechanism. Consequently, with three items, the optimal mechanism induces a partial order among the organic items, as outlined in Theorem \ref{thm:partial order}. Moreover, we can demonstrate that the bids of items lower in this order can influence the allocation of items higher in the order, while the reverse is not true. 


\begin{theorem}\label{thm:partial order}
    In the optimal 2-out-of-3 mechanism, 
    there exists a partial order $(i_1, i_2, i_3)$ such that \(y_{i_1}^* \geq y_{i_2}^* \geq y_{i_3}^*\), and \(y_{i_1}^* = y_{i_1}^*(a_{i_2}, a_{i_3})\), \(y_{i_2}^* = y_{i_2}^*(a_{i_3})\), \(y_{i_3}^* = 0\).
\end{theorem}

    
    
    

\section{Generalized Fix Mechanism}
\label{sec:fix}
In the last section, we provide a more detailed characterization of the optimal merging mechanisms, but we find that designing the optimal mechanism is not trivial.  Intuitively, one might think that ranking all \(a_i\) and \(o_i\) in non-increasing order and selecting the \(k\) items with the highest values, while avoiding conflicts in form, would be sufficient based on the objective function in the merging problem (\ref{fml:optimization 2}). Yet, due to a key feature of our model—where displaying an item in its organic form yields greater utility—this approach may not be incentive compatible. Some items could manipulate their displayed form from ad to organic by misreporting, aiming for better utility. Thus, a Myerson-like mechanism cannot be directly applied to our problem.



Starting from this section, we introduce two simple merging mechanisms and analyze their performance relative to the optimal mechanism. In this section, we first introduce the generalized fix (G-FIX) mechanism  and then present the generalized change mechanism in next section. The primary challenge in designing 
a merging mechanism lies in determining the displayed form of each item. To address this, we restrict some items to be displayed only in their organic form, while others are restricted to ad form. We then select the \(k\) items with the highest contribution, leading to the G-FIX-\(I\) mechanism.



\begin{definition}[G-FIX-$I$ Mechanism]
Define a set $I\subseteq [n]$ such that  items in $I$ can only be displayed in organic form, while items in $[n] \setminus I$ can only be displayed in ad form. The G-FIX-I mechanism, denoted by $\mathcal{M}_I^F$, is defined as follows: Let $L$ be   
    \[ 
         L = \{o_i : i \in I\} \cup \{a_i : i \in  [n] \setminus I\} .
    \]
Then, for each item \(i \in [n]\):
    \[
        x_i(\boldsymbol{a}) = 
            \begin{cases} 
            1 & \text{if } i \notin I \text{ and } a_i > {\max}^{(k+1)}L -{\max}^{(k)}L ,\\ 
            0 & \text{otherwise}.
            \end{cases}
    \]
    \[
        y_i(\boldsymbol{a}) = 
        \begin{cases} 
        1 & \text{if } i \in I \text{ and } o_i > {\max}^{(k+1)}L -{\max}^{(k)}L,\\ 
        0 & \text{otherwise}.
        \end{cases}
    \]
\end{definition}
The objective value of \(\mathcal{M}_I^F\) is then calculated by:
\begin{small}
    \[
    \operatorname{OBJ}(\mathcal{M}_I^F) = \underset{\boldsymbol{a} \sim G }{\mathbb{E}}\left[{\max}^{(k)} {\{o_i : i \in I\} \cup \{a_i : i \in [n] \setminus I\}}\right].
    \]
\end{small}


It is easy to verify that the constraints in the merging problem (\ref{fml:optimization 2}) are satisfied by the allocation rules of the G-FIX-$I$ mechanism, ensuring that $\mathcal{M}_I^F$ is IC, IR, and feasible. By evaluating all \(\mathcal{M}_I^F\) mechanisms with \(|I| \leq k\) and selecting the one with the highest objective value, we obtain our first merging mechanism, the G-FIX mechanism.



\begin{definition}[G-FIX Mechanism]\label{def:G-Fix Mechanism}
    The G-FIX mechanism, denoted by \( \mathcal{M}^F \), selects the mechanism with the highest objective value among all \( \mathcal{M}_I^F \) mechanisms, where \( I \subseteq [n] \) and \(|I| \leq k\).
\end{definition}

We next compare the performance of the G-FIX mechanism to the optimal merging mechanism, starting with the simplest setting of 3 items and 2 slots. 




\begin{theorem}\label{thm: fix 2 from 3}
For content merging problem with $3$ items and $2$ slots, $\mathcal{M}^{F}$ is $\left({4}/{5}\right)^{3}$-approximate relative to the optimal mechanism.
\end{theorem}

Returning to the general setting with \(n\) items and \(k\) slots, we also demonstrate that \(\mathcal{M}^F\) has a guaranteed approximate ratio under certain assumptions. 



\begin{theorem}\label{thm: fix k from n}
For content merging problem with $n$ items and $k$ slots, when $2k \leq n$ and the bid distributions are identical, $\mathcal{M}^{F}$ is $C_{n-k}^{k}/C_{n}^{k}$-approximate mechanism. 
\end{theorem}
\begin{proof}[Proof Sketch]
Without loss of generality, we assume that $o_1 \geq o_2$ $ \geq \cdots \geq o_n$ in this proof. Before giving the proof of this theorem, we first introduce a mathematical lemma to aid us analyze the approximate ratio. 
\begin{lemma}\label{lem:c_(n-k)/c_n}
When $1 \leq l \leq k$, we have
$$
\frac{1-\sum\limits_{i=1}^lC_{n-k}^{l-i}x^{(n-k)-(l-i)}(1-x)^{(l-i)}}{1-\sum\limits_{i=1}^lC_{n}^{l-i}x^{n-(l-i)}(1-x)^{(l-i)}} \geq \frac{C_{n-k}^{l}}{C_{n}^{l}}.
$$
\end{lemma}
Firstly, observe the objective of the merging problem (\ref{fml:optimization 2}). We know that for the optimal mechanism $\mathcal{M}^*$, its objective is always bounded by ${\max}^{(k)}\{o_1\cdots o_k, a_1\cdots a_n\}$. Denote by \(\lambda_i\) the \(i\)-th largest number in the set $\{o_1\cdots o_k,a_1\cdots a_n\}$. We have
\begin{align}
    \operatorname{OBJ}(\mathcal{M}^{*})  \leq& \underset{\boldsymbol{a}\sim G}{\mathbb{E}}\left[{\max}^{(k)}\{o_{1} \cdots o_{k}, a_{1} \cdots a_{n}\}\right] \nonumber\\
    =&\sum\limits_{i=1}^{k}\int_{0}^{\infty}\underset{\boldsymbol{a}\sim G}{\mathbb{P}}\left( \lambda_i\geq t\right)dt  \label{ieq: fix k from n upper on optimal}.
\end{align}

On the other hand, we want to find a lower bound on the objective of $\mathcal{M}^F$. Since $\mathcal{M}^F$ selects the best mechanism among $\mathcal{M}_I^F$, herein, we focus on a specific $\mathcal{M}_I^F$ with $I=\{1,2,\cdots, k\}$, whose objective is ${\max}^{(k)}\{o_1,$ $\cdots o_k,a_{k+1}\cdots a_n\}$. Denote by \(\theta_i\) the \(i\)-th largest number in the set $\{o_1\cdots o_k,a_{k+1}\cdots a_n\}$. We have
\begin{align}
     \operatorname{OBJ}(\mathcal{M}^{F})  \geq &\underset{a\sim G}{\mathbb{E}}\left[{\max}^{(k)}\{o_{1} \cdots o_{k}, a_{k+1} \cdots a_{n}\}\right] \notag\\
     =&\sum\limits_{i=1}^{k}\int_{0}^{\infty}\underset{\boldsymbol{a}\sim G}{\mathbb{P}}( \theta_i\geq t)dt.\label{ieq: fix k from n lower on fix}
\end{align}
With a calculation on the probability of order statistics, and by Lemma \ref{lem:c_(n-k)/c_n} and the following inequality 
\[
    1 = \frac{C_{n-k}^{0}}{C_{n}^{0}}>\cdots>\frac{C_{n-k}^{l}}{C_{n}^{l}}>\frac{C_{n-k}^{l+1}}{C_{n}^{l+1}}>\cdots>\frac{C_{n-k}^{k}}{C_{n}^{k}},
\]
we can replace the simplified terms in inequality (\ref{ieq: fix k from n lower on fix}) by the corresponding  simplified terms in inequality (\ref{ieq: fix k from n upper on optimal}) with ${C_{n-k}^{k}}/{C_{n}^{k}}$ times and we get
\begin{align*}
    \operatorname{OBJ}(\mathcal{M}^{F})  \geq 
    \frac{C_{n-k}^{k}}{C_{n}^{k}}\operatorname{OBJ}(\mathcal{M}^{*}).
\end{align*}


In conclusion, we show that $\mathcal{M}^{F}$ is ${C_{n-k}^{k}}/{C_{n}^{k}}$-approximate relative to the optimal mechanism $\mathcal{M}^{*}$.
\end{proof}
Based on Theorem \ref{thm: fix k from n}, if \(n \geq k^2/\epsilon + k\) for any given \(0 < \epsilon < 1\)—\ that is, if the number of candidate items \(n\) is sufficiently large compared to the slots \(k\)—we can further show that \(\mathcal{M}^F\) is a near-optimal mechanism.


\begin{theorem}
    For content merging problem with $n$ items and $k$ slots, given any $0<\epsilon<1$, when $n\geq k^2/\epsilon +k$ and the bid distributions are identical, $\mathcal{M}^{F}$ is a nearly-optimal mechanism with approximate ratio $1-\epsilon$.
\end{theorem}

\begin{proof}
We directly analyze the approximate ratio of the mechanism $\mathcal{M}^{F}$, presented in Theorem \ref{thm: fix k from n}, i.e., 
    \begin{align*}
        \frac{C_{n-k}^k}{C_n^k} 
        &=\frac{n-k}{n}\cdot\frac{n-k-1}{n-1}\cdots\frac{n-2k+1}{n-k+1}\\
        &>\left(\frac{n-2k}{n-k}\right)^k = \left(1-\frac{k}{n-k}\right)^k.
    \end{align*}

    Applying the general inequality $(1+x)^a\geq 1+ax$ where $a>1$, we obtain: 
\[
    \left(1 - \frac{k}{n-k}\right)^k \geq 1 - \frac{k^2}{n-k}.
\]
Therefore, to ensure that the approximate ratio is greater than $1 - \epsilon$, it is sufficient to require:
\begin{align*}
    \frac{k^2}{n-k} \leq \epsilon 
    \implies n \geq \frac{k^2}{\epsilon} + k. 
\end{align*}
This completes the proof.
\end{proof}

\section{Generalized Change Mechanism}\label{sec:change}
In the last section, we introduce the G-FIX mechanism, which restricts each owner to participate in the mechanism with a specific form. However, it lacks flexibility and results in a loss of objective value. In this section, we relax this constraint by incorporating an order factor, leading to the establishment of the generalized change (G-CHANGE) mechanism.



We begin by introducing the generalized change-$I$ (G-CHANGE-$I$) mechanism. The G-CHANGE-$I$ mechanism defines an ordered set \(I\) in which items can be displayed either as ads or in their organic form, while all other items are restricted to ads only. The final display forms are determined by a recursive rule, which we will detail later. Compared to the G-FIX mechanism, the G-CHANGE mechanism offers more flexibility by allowing some items to have alternative forms and introducing a criterion to guide the decision-making process. The detailed definition is as follows.



\begin{definition}[Generalized Change-$I$ Mechanism]\label{def:change_weight}
Let $I = \{i_1, i_2, \cdots, i_k\}$ be a set of ordered items, where $k \leq n$. 
Given the bid profile $\boldsymbol{a}$, define $w_{i_t}^{\boldsymbol{a}}, t\in [k]$, recursively as follows:
\begin{small}
$$
\begin{cases}
    w_{i_1}^{\boldsymbol{a}} = \underset{{a_{i_1}}}{\mathbb{E}}[\max^{(k)}\{a_1, \cdots, a_n \}] - R_1&t = 1, \\
    w_{i_t}^{\boldsymbol{a}} = \underset{{a_{i_t}}}{\mathbb{E}}[\max \{o_{t-1} + R_{t-1}, w_{i_{t-1}}^{\boldsymbol{a}} + R_{t-1}\}] - R_t &t > 1,
\end{cases}
$$
\end{small}
where
\[
    R_t = \sum_{j=1}^{t-1} o_{i_j} +  {\max}^{(k-t)}\left\{a_i \mid i \notin \{i_1, \cdots, i_t\}\right\}.
\]
Define a threshold $s^*$ as
\[
    s^* := \max \left\{ s\in[k] \mid \forall t > s, \, o_{i_t} < w_{i_t}^{\boldsymbol{a}} \text{ and } o_{i_s} \geq w_{i_s}^{\boldsymbol{a}} \right\}.
\]
Having the above symbols, generalized change-$I$ mechanism is given by
\begin{small}
\[
\begin{aligned}
    x_i(\boldsymbol{a}) &= 
    \begin{cases} 
    1 & \text{if } i \notin \{i_1, \cdots, i_{s^*}\} \\
    & \text{and } a_i >{\max}^{(k-s^*+1)} \{ a_j \mid j \notin \{i_1, \cdots, i_{s^*}\} \}\\
    &\quad\quad - {\max}^{(k-s^*)} \{ a_j \mid j \notin \{i_1, \cdots, i_{s^*}\} \} ,\\ 
    0 & \text{otherwise} .
    \end{cases} \\
    y_i(\boldsymbol{a}) &= 
    \begin{cases} 
    1 & \text{if } i \in \{i_1, \cdots, i_{s^*}\} ,\\ 
    0 & \text{otherwise} .
    \end{cases}
    \end{aligned}
\]
\end{small}
\end{definition}
Herein, we elaborate on the decision rule in the G-CHANGE-$I$ mechanism. For the items in $I$, we evaluate them in reverse order. Take item $i_t$ as an example. $R_t$ represents the total contribution when the first $t-1$ items in $I$ are displayed in organic form, and the remaining $k-t$ items with the highest $a_i$ values are selected in ad form. $w_{i_t}$ can be interpreted as the expected marginal contribution of item $i_t$'s ad form, relative to the local optimal value. Thus, the display form of item \(i_t\) is determined by comparing \(w_{i_t}\) with \(o_{i_t}\) in reverse order, identifying the first item for which \(o_{i_t} \geq w_{i_t}\), denoted by \(s^*\). Finally, the mechanism displays the first \(s^*\) items in \(I\) in organic form and selects \(k-s^*\) ads from the remaining items based on their \(a_i\) values.

It is not hard to find that the decision rule is independent of the items' bids. Additionally, the reverse decision also guarantees that flexible items in $I$ cannot affect bilaterally. These two features indicates that the G-CHANGE-$I$ mechanism satisfies the IC condition. Furthermore, the objective function of the mechanism $\mathcal M_{I}^{C}$ can be expressed as:
\begin{equation*}\label{eqn:obj change I}
    \text{OBJ}(\mathcal M_{I}^{C}) = \underset{\boldsymbol{a} \sim G}{\mathbb{E}} \left[ \max \left\{ \sum_{j=1}^{k} o_{i_j}, w_{i_k}^{\boldsymbol{a}} + \sum_{j=1}^{k-1} o_{i_j} \right\} \right].
\end{equation*}
We can evaluate all possible ordered subset \(I\) and select the one with the highest objective value as the generalized change mechanism:

\begin{definition}[Generalized Change Mechanism]\label{def:change mechanism}
The generalized change mechanism, denoted by $\mathcal M^{C}$, executes the mechanism with the highest objective value among $\mathcal M_{I}^{C}$, where $I$ is any ordered subset (permutation) satisfying $I\subseteq [n]$ and $|I|=k$.
\end{definition}
Next, we examine the performance of the G-CHANGE mechanism in the simplest setting.

\begin{theorem}\label{thm: change 2 from 3}
For content merging problem with $3$ items and $2$ slots, $\mathcal{M}^{C}$ is optimal.
\end{theorem}
\begin{proof}
By Theorem $\ref{thm:partial order}$, we know that 
there exists a partial order $(i_1, i_2, i_3)$, which is a permutation of $\{1, 2, 3\}$, such that $y_{i_1}^* \geq y_{i_2}^* \geq y_{i_3}^* = 0$, and $ y_{i_1}^* = y_{i_1}^*(a_{i_2}, a_{i_3}), y_{i_2}^* = y_{i_2}^*(a_{i_3}), y_{i_3}^* = 0$. Therefore, the objective value of the optimal mechanism $\mathcal M^*$ is given by:
\begin{flalign*}
    &\text{OBJ}(\mathcal M^*) = \underset{\boldsymbol{a} \sim G}{\mathbb{E}}\left[\sum_{i=1}^3 (x_i^*(a) a_i + y_i^*(a) o_i)\right] &
\end{flalign*}
\begin{align*}
    &=\underset{a_{i_3}}{\mathbb{E}}\Bigg[\max \Big\{o_{i_1} + o_{i_2},\underset{a_{i_2}}{\mathbb{E}}\bigg[\max \big\{ o_{i_1} + \max\{a_{i_2},a_{i_3}\}, \\
    &\underset{a_{i_1}}{\mathbb{E}}[{\max}^{(2)} \{a_{i_1}, a_{i_2}, a_{i_3}\}]\big\}\bigg]\Big\}\Bigg].
\end{align*}
Observing from the objective of G-CHANGE-$I$, we can find that if we let $I=\{i_1,i_2\}$, it holds that 
$\text{OBJ}(\mathcal M^*) = \text{OBJ}(\mathcal M_{\{i_1, i_2\}}^{C})$. On the other hand, we know that  $ \text{OBJ}(\mathcal M_{\{i_1, i_2\}}^{C}) \leq \text{OBJ}(\mathcal M^C)$. Thus, the G-CHANGE mechanism $\mathcal M^C$ is optimal in the simplest setting.
\end{proof}

Finally, in the general case with non-identical distributions, the G-CHANGE mechanism still achieves guaranteed performance relative to the optimal mechanism.

\begin{theorem}
For content merging problem with $n$ items and $k$ slots, $\mathcal{M}^{C}$ is $1/2$-approximate relative to the optimal mechanism $\mathcal{M}^{*}$.
\end{theorem}

\begin{proof}
Without loss of generality, let \( o_1 \geq o_2 \geq \cdots \geq o_n \). Firstly, the objective value of the optimal mechanism is upper bounded by
\begin{align*}
\text{OBJ}(\mathcal M^*) &\leq \underset{\boldsymbol{a}\sim G}{\mathbb{E}}\left[ {\max}^{(k)} \{ o_1 \cdots o_k, a_1 \cdots a_n \} \right] \\
&\leq \sum_{i=1}^k o_i + \underset{\boldsymbol{a} \sim G}{\mathbb{E}}\left[ {\max}^{(k)} \{ a_1, \cdots, a_n \}\right].
\end{align*}
Let $I=\{1,\cdots,k\}$ and we can get a low bound of $\mathcal M^C$ 
\begin{align*}
\text{OBJ}(\mathcal M^C) &\geq \text{OBJ}(M_{\{1, \cdots, k\}}^C) \\
&\geq \max \left\{\sum_{i=1}^k o_i , \underset{\boldsymbol{a} \sim G}{\mathbb{E}}\left[{\max}^{(k)} \{ a_1, \cdots, a_n \}\right]\right\}.
\end{align*}
Thus, by triangle inequality, we obtain 
\begin{align*}
\frac{\text{OBJ}(\mathcal M^C)}{\text{OBJ}(\mathcal M^*)} &\geq \frac{\max \left\{\sum\limits_{i=1}^k o_i , \underset{\boldsymbol{a}\sim G}{\mathbb{E}}\left[{\max}^{(k)} \{ a_1, \cdots, a_n \}\right]\right\}}{\sum\limits_{i=1}^k o_i + \underset{\boldsymbol{a}\sim G}{\mathbb{E}}\left[ {\max}^{(k)}\{ a_1, \cdots, a_n \} \right]}\\
& \geq  \frac{1}{2}. 
\end{align*}
\end{proof}

\section{Conclusion}
\label{sec:conclusion}

In this paper, we tackle the merging mechanism design problem in multiple slots, addressing an open question from \cite{LMZWZYXZD24}, where items can be displayed as either ads or in organic form. These different forms impact platform performance outcomes, such as click-through rates and user experience. We propose two straightforward mechanisms that satisfy the incentive compatibility, individual rationality, and feasibility conditions, and analyze their approximation ratio in both simple and general settings. Future research directions include: first, designing the optimal merging mechanisms or improving approximation ratio within this framework; and second, exploring the merging problem in heterogeneous slots, which presents additional challenges.


\clearpage
\newpage
\section*{Acknowledgments}
This work was supported by National Natural Science Foundation of China No. 62472428, the Fundamental Research Funds for the Central Universities, and the Research Funds of Renmin University of China (No. 22XNKJ07, No. 23XNH028), CCF-Huawei Populus Grove Fund (No. CCF-HuaweiLK2023005), and Major Innovation \& Planning Interdisciplinary Platform for the “Double-First Class” Initiative, Renmin University of China, and the Fundamental Research Funds of Shandong University (No.11480061330000).

\bibliography{aaai25}
\clearpage
\appendix 
\appendix

\section*{Appendix}\label{sec:app}

\section{Missing Proofs in Section \ref{sec:fs}}
\label{app:proofs in fs}
\subsection{Proof of Proposition \ref{prop: k from n optimal}}

\begin{proof}
To analyze how changing the bid of a single item affects the objective function, for any \(i \in [n]\), \(a_i\) and \(\boldsymbol{a}_{-i}\), we can decompose the objective function in problem (\ref{fml:optimization 2}) as follows:
\begin{align}
    &\max\limits_{x,y} \quad \operatorname{OBJ}(\mathcal{M}^{x,y}) = \underset{\boldsymbol{a} \sim G}{\mathbb{E}}\left[ \sum\limits_{j=1}^{n}(a_j x_j(\boldsymbol{a}) + o_jy_j(\boldsymbol{a})) \right] \nonumber \\
    &= \underset{\boldsymbol{a}}{\mathbb{E}}\left[\mathbb{I}[\boldsymbol{a}_{-i}] \mathbb{I}[a_i' \leq a_i] \left(\sum\limits_{j=1}^{n} a_j x_j(\boldsymbol{a}) + o_j y_j(\boldsymbol{a}) \right)\right] \tag{$\mathrm{s}_1$} \\
    &+ \underset{\boldsymbol{a}}{\mathbb{E}}\left[\mathbb{I}[\boldsymbol{a}_{-i}] \mathbb{I}[a_i' > a_i] \left(\sum\limits_{j=1}^{n} a_j x_j(\boldsymbol{a}) + o_j y_j(\boldsymbol{a}) \right)\right] \tag{$\mathrm{s}_2$} \\
    &+ \underset{\boldsymbol{a}}{\mathbb{E}}\left[\left(1 - \mathbb{I}[\boldsymbol{a}_{-i}]\right) \left(\sum\limits_{j=1}^{n} a_j x_j(\boldsymbol{a}) + o_j y_j(\boldsymbol{a}) \right)\right] \tag{$\mathrm{s}_3$}.
\end{align}
Here, $\mathrm{s}_1$ represents the scenario where the values of all  $\mathrm{Ad}$s except for item \( i \), are given as \(\boldsymbol{a}_{-i} \) and the value of $\mathrm{Ad}_{i}$ is \( a_i' \leq a_i \). Therefore, we can rewrite this as:
\[
    \mathbb{P}(\boldsymbol{a}_{-i})\underset{a_i'}{\mathbb{E}}\left[\mathbb{I}[a_i' \leq a_i] \sum_{j=1}^n \left(a_j x_j(a_i', \boldsymbol{a}_{-i}) + o_j y_j(a_i',\boldsymbol{a}_{-i})\right)\right].
\]

Similarly, $\mathrm{s}_2$ represents the scenario where the values of all  $\mathrm{Ad}$s except for item \( i \), are given as \(\boldsymbol{a}_{-i} \) and the value of $\mathrm{Ad}_{i}$ is \( a_i' > a_i \). Therefore, we can rewrite this as:
\[
    \mathbb{P}(\boldsymbol{a}_{-i})\underset{a_i'}{\mathbb{E}}\left[\mathbb{I}[a_i '> a_i] \sum_{j=1}^n (a_j x_j(a_i', \boldsymbol{a}_{-i}) + o_j y_j(a_i',\boldsymbol{a}_{-i}))\right].
\]

Part \( \mathrm{s}_3 \)  represents the case where the values of the remaining items are not equal to $\boldsymbol{a}_{-i}$.

\subsubsection{Statement 1:}
Recall Statement 1 in Proposition \ref{prop: k from n optimal}, and we only need to focus on \( \mathrm{s}_1 \) and \( \mathrm{s}_2 \) and optimize the corresponding objective:
\[
    \underset{a_i'}{\mathbb{E}} \left[\sum_{j=1}^n (a_j x_j(a_i', \boldsymbol{a}_{-i}) + o_j y_j(a_i', \boldsymbol{a}_{-i}))\right].
\]
Since $y_i^*(a_i, \boldsymbol{a}_{-i}) = 1$, by form stability from Lemma \ref{lem:original form stability}, we know that for any \( a_i' \),  \( y_i(a_i', \boldsymbol{a}_{-i}) = 1 \) and \( x_i(a_i', \boldsymbol{a}_{-i}) = 0 \). Therefore, the equation above can be rewritten as: 
\[
    o_i + \underset{a_i'}{\mathbb{E}} \left[\sum_{j \ne i} (a_j x_j(a_i', \boldsymbol{a}_{-i}) + o_j y_j(a_i', \boldsymbol{a}_{-i}))\right].
\]
Since \( a_j \) and \( o_j \) (\( j \ne i \)) remain unchanged in this scenario, the optimal mechanism's allocation for the remaining items also remains the same. That is, for any \( j \) and any \( a_i' \), we have \( y_j^*(a_i', \boldsymbol{a}_{-i}) = y_j^*(a_i, \boldsymbol{a}_{-i}) \) and \( x_j^*(a_i', \boldsymbol{a}_{-i}) = x_j^*(a_i, \boldsymbol{a}_{-i}) \).

\subsubsection{Statement 2:}
To demonstrate Statement 2 in the proposition, we focus on part $\mathrm{s}_1$ and optimize the following objective:
\[
    \underset{a_i'}{\mathbb{E}} \left[\mathbb{I}[a_i' \leq a_i] \left( \sum_{j=1}^n (a_j x_j(a_i', \boldsymbol{a}_{-i}) + o_j y_j(a_i', \boldsymbol{a}_{-i})\right)
    \right].
\]
By form stability from Lemma \ref{lem:original form stability} and allocation monotonicity from Lemma \ref{lem:myerson}, we know that for any $a_i'\leq a_i$, we have $y_i(a_i',\boldsymbol{a}_{-i}) = y_i(a_i, \boldsymbol{a}_{-i})$  and $ x_i(a_i', \boldsymbol{a}_{-i}) = 0$. Therefore, the equation above can be rewritten as:
\begin{align*}
    &o_i y_i(a_i', \boldsymbol{a}_{-i}) +\\
    &\underset{a_i'}{\mathbb{E}}\left[\mathbb{I}[a_i' < a_i] \left(\sum_{j \ne i} (a_j x_j(a_i', \boldsymbol{a}_{-i}) + o_j y_j(a_i', \boldsymbol{a}_{-i})\right)\right].
\end{align*}

Since $a_j$ and $o_j ( j \ne i )$ remain unchanged in this scenario, the optimal mechanism's allocation for the remaining items does not change. That is, for any $j$ and any $a_i'$, we have $ y_j^*(a_i', \boldsymbol{a}_{-i}) = y_j^*(a_i, \boldsymbol{a}_{-i})$ and $x_j^*(a_i', \boldsymbol{a}_{-i}) = x_j^*(a_i, \boldsymbol{a}_{-i}) $.

\subsubsection{Statement 3:}
Similar to the argument in Statement 2, we concentrate on part $s_2$ and optimize the objective 
\begin{align*}
    &\underset{a_i'}{\mathbb{E}}\left[\mathbb{I}\left[a_i' > a_i\right]a_i'\right] + \\
    &\underset{a_i'}{\mathbb{E}}\left[
        \mathbb{I}[a_i'> a_i] \left(\sum_{j \ne i} (a_j x_j(a_i', \boldsymbol{a}_{-i}) + o_j y_j(a_i', \boldsymbol{a}_{-i})\right)
    \right].
\end{align*}
Due to form stability, allocation monotonicity, and the optimality, we conclude that for any \( j \) and any \( a_i' >a_i\), it holds that \( y_j^*(a_i', \boldsymbol{a}_{-i}) = y_j^*(a_i, \boldsymbol{a}_{-i}) \) and \( x_j^*(a_i', \boldsymbol{a}_{-i}) = x_j^*(a_i, \boldsymbol{a}_{-i}) \).

In summary, we have shown this proposition.
\end{proof}

\subsection{Proof of Corollary \ref{cor:optimal form stability}}
\begin{proof}
    In an optimal mechanism that satisfies IC, IR, and feasibility, if the display form of item $i$ remains unchanged, i.e., \(X_i^*(a_i', \boldsymbol{a}_{-i}) = X_i^*(a_i, \boldsymbol{a}_{-i})\), this corresponds precisely to one of the three cases outlined in Proposition \ref{prop: k from n optimal}. Consequently, Proposition \ref{prop: k from n optimal} directly implies that the display forms of the other items also remain unchanged. 
\end{proof}

\subsection{Proof of Lemma \ref{lem:2 from 3 fs}}
\begin{proof}

    Without loss of generality, let \(i=1\) and \(j=2\). Suppose there exist \(\boldsymbol{a} = (a_1, a_2, a_3)\) and \(\boldsymbol{a'} = (a_1', a_2', a_3')\) such that \(y_1^*(a_1, a_2, a_3) = 1\), \(y_2^*(a_1, a_2, a_3) = 0\), but \(y_1^*(a_1', a_2', a_3') = 0\), \(y_2^*(a_1', a_2', a_3') = 1\).
    
    According to Statement 1 in Proposition \ref{prop: k from n optimal}, which states that when \(y_i^* = 1\), a change in \(a_i\) alone does not affect the optimal allocation of all items, we can derive that changing $a_1$ to $a'_1$ in profile $\boldsymbol{a}$ leads to \(y_1^*(a_1', a_2, a_3) = 1\) and  \(y_2^*(a_1', a_2, a_3) = 0\). On the other hand,  but changing $a'_2$ to $a_2$ in profile $\boldsymbol{a}'$ leads to \(y_1^*(a_1', a_2, a_3') = 0\) and  \(y_2^*(a_1', a_2, a_3') = 1\).
    
    This implies that changing the bid of item 3 alone causes the allocations of items 1 and 2 to change. According to Corollary \ref{cor:optimal form stability}, this change also indicates that the allocation form of item 3 must itself change. 
    
    Without loss of generality, let \(a_3' > a_3\). Since there are only two slots, we observe that before and after the change, the corresponding allocations must be \(y_1^*(a_1', a_2, a_3) = 1\), \(x_2^*(a_1', a_2, a_3) = 1\), $x_3^*(a_1', a_2, a_3) = 0$ and \(y_2^*(a_1', a_2, a_3') = 1\), \(x_3^*(a_1', a_2, a_3') = 1\).
    
    Since item 2 is displayed in the ad form under the profile  \((a_1', a_2, a_3)\), we know 
    \[
        \underset{a_2}{\mathbb{E}}[{\max}^{(2)}\{o_1, a_2, a_3\}] > {\max}^{(2)}\{o_1, o_2, a_3\}.
    \]
    Because $\text{Org}_1$ is ultimately selected, i.e., \(o_1 > a_3\), it follows that 
    \[
        o_1 + \underset{a_2}{\mathbb{E}}\left[\max\{a_2, a_3\}\right] > o_1 + \max\{o_2, a_3\},
    \]
    and consequently 
    \[
        \underset{a_2}{\mathbb{E}}\left[\max\{a_2, a_3\}\right] > \max\{o_2, a_3\} \ge o_2.
    \]
    
    On the other hand, since item 2 is displayed in the organic form under the profile  \((a_1', a_2, a_3')\), we know 
    \begin{align*}
            o_2 + a_3' &> \underset{a_2}{\mathbb{E}}[\sum_{i=1}^3 x_i^*(\boldsymbol{a}) a_i + y_1^*(\boldsymbol{a}) o_1+y_3^*(\boldsymbol{a}) o_3]\\
            &\ge \underset{a_2}{\mathbb{E}}[{\max}^{(2)}\{o_1, a_2, a_3'\}] \\
            &\ge a_3' + \underset{a_2}{\mathbb{E}}\left[\max\{o_1, a_2\}\right], 
    \end{align*}
    which implies a contradiction:
    \[
    o_2 > \underset{a_2}{\mathbb{E}}\left[\max\{o_1, a_2\}\right] > \underset{a_2}{\mathbb{E}}\left[\max\{a_2, a_3\}\right] > o_2.
    \]
    
    Therefore, in an optimal $2$-out-of-$3$ mechanism that satisfies IC, IR, and feasibility, if there exists some \(\boldsymbol{a}\) such that \(y_i^*(\boldsymbol{a})=1\) and \(y_j^*(\boldsymbol{a})=0\), then for any \(\boldsymbol{a'}\), it is not possible for \(y_i^*(\boldsymbol{a}')=0\) and \(y_j^*(\boldsymbol{a}')=1\).
    
\end{proof}

\subsection{Proof of Theorem \ref{thm:partial order}}
\begin{proof}
    Firstly, based on Lemma \ref{lem:2 from 3 fs}, for any two items \(i\) and \(j\), either \(y_i^* \geq y_j^*\) or \(y_j^* \geq y_i^*\). This implies that there is a total ordering among the three items, resulting in a partial order such that \(y_{i_1}^* \geq y_{i_2}^* \geq y_{i_3}^*\), where $i_1$, $i_2$, $i_3$ is a permutation of the three items.

    It is trivial to see that \(y_{i_3}^* = 0\); otherwise, we would have \(y_{i_1}^* = y_{i_2}^* = y_{i_3}^* = 1\), which violates the feasibility constraint of selecting only two items.
    
    Next, we consider the effect of changing \(a_{i_1}\) on \(y_{i_2}^*\). When \(y_{i_1}^* = 1\), Proposition \ref{prop: k from n optimal} tells us that changes in \(a_{i_1}\) do not affect \(y_{i_2}^*\).When \(y_{i_1}^* = 0\), the partial order constraint implies that \(y_{i_2}^* = 0\). Therefore, \(y_{i_2}^*\) is not affected by changes in \(a_{i_1}\).
    
    Additionally, by form stability from Lemma \ref{lem:original form stability}, changes in \(a_i\) alone do not affect \(y_i\). Combining these results, we conclude that:
    \begin{align*}
        &y_{i_1}^* = y_{i_1}^*(a_{i_2}, a_{i_3}),\\
        &y_{i_2}^* = y_{i_2}^*(a_{i_3}), \\
        &y_{i_3}^* = 0.
    \end{align*}
    Thus, we have shown that there exists a partial order \((i_1, i_2, i_3)\) such that \(y_{i_1}^* \geq y_{i_2}^* \geq y_{i_3}^*\), with \(y_{i_1}^* = y_{i_1}^*(a_{i_2}, a_{i_3})\), \(y_{i_2}^* = y_{i_2}^*(a_{i_3})\), and \(y_{i_3}^* = 0\).
\end{proof}

\section{Missing Proofs in Section \ref{sec:fix}}
\subsection{Proof of Theorem \ref{thm: fix 2 from 3}}
\begin{proof}
When $n=3$ and $k=2$, the mechanism $\mathcal{M}^F$ performs best among the following: $\mathcal{M}_0$, $\mathcal{M}_{\{1\}}^F$, $\mathcal{M}_{\{2\}}^F$, $\mathcal{M}_{\{3\}}^F$, $\mathcal{M}_{\{1,2\}}^F$, $\mathcal{M}_{\{1,3\}}^F$ and $\mathcal{M}_{\{2,3\}}^F$. We have the following result:
\begin{align*}
    \operatorname{OBJ}(\mathcal{M}_0) &= \underset{\boldsymbol{a}\sim G}{\mathbb{E}}\left[{\max}^{(2)}\{a_1,a_2,a_3\}\right], \\
    \operatorname{OBJ}(\mathcal{M}_{\{1\}}^F) &= \underset{a_2,a_3}{\mathbb{E}}\left[{\max}^{(2)}\{o_1,a_2,a_3\}\right], \\
    \operatorname{OBJ}(\mathcal{M}_{\{2\}}^F) &= \underset{a_1,a_3}{\mathbb{E}}\left[{\max}^{(2)}\{a_1,o_2,a_3\}\right], \\
    \operatorname{OBJ}(\mathcal{M}_{\{3\}}^F) &= \underset{a_1,a_2}{\mathbb{E}}\left[{\max}^{(2)}\{a_1,a_2,o_3\}\right], \\
    \operatorname{OBJ}(\mathcal{M}_{\{1,2\}}^F) &= \underset{a_3}{\mathbb{E}}\left[{\max}^{(2)}\{o_1,o_2,a_3\}\right], \\
    \operatorname{OBJ}(\mathcal{M}_{\{1,3\}}^F) &= \underset{a_2}{\mathbb{E}}\left[{\max}^{(2)}\{o_1,a_2,o_3\}\right], \\
    \operatorname{OBJ}(\mathcal{M}_{\{2,3\}}^F) &= \underset{a_1}{\mathbb{E}}\left[{\max}^{(2)}\{a_1,o_2,o_3\}\right].
\end{align*}

To analyze the performance of $\mathcal{M}^F$, we first establish an upper bound for the optimal mechanism $\mathcal{M^*}$. Based on Theorem \ref{thm:partial order}, there exists a partial order relationship in the 2-out-of-3 optimal mechanism. Without loss of generality, assume \(y_{1}^{*} \geq y_{2}^{*} \geq y_{3}^{*}\). Let \(y_{1}^{*} = y_{1}^{*}(a_{2}, a_{3})\), \(y_{2}^{*} = y_{2}^{*}(a_{3})\), and \(y_{3}^{*} \equiv 0\).
Consequently, the objective function can be rewritten as follows:
\begin{align*}
    &\operatorname{OBJ}(\mathcal{M}^{*}) = \underset{\boldsymbol{a}}{\mathbb{E}}\left[\sum\limits_{i=1}^{3} y_{i}^{*}(\boldsymbol{a}) o_{i}+\sum\limits_{i=1}^{3} x_{i}^{*}(\boldsymbol{a}) a_{i}\right]\\
    &=\underset{a_3}{\mathbb{E}}\left[y_2^*(a_3) o_2+\underset{a_2}{\mathbb{E}}\left[y_1^*(a_2,a_3) o_1 + \underset{a_1}{\mathbb{E}}\left[\sum\limits_{i=1}^{3} x_i^*(\boldsymbol{a})a_i\right]\right]\right].
\end{align*}

We simplify this objective function by examining it from the outside in. Fixing any $a_3$, we consider two cases:
\begin{enumerate}[(1)]
    \item If \( y_2^*(a_3) = 1\) for all $a_1$ and $a_2$, then by the partial order relationship, we have \( y_1^* \geq y_2^* = 1 \). Consequently, the mechanism selects two organic items. Thus, we have:
    \begin{align*}
        &y_2^*(a_3) o_2+\underset{a_2}{\mathbb{E}}\left[y_1^*(a_2, a_3) o_1 + \underset{a_1}{\mathbb{E}}\left[\sum\limits_{i=1}^{3} x_i^*(\boldsymbol{a})a_i\right]\right] \\
        =& o_1+o_2
        \leq {\max}^{(2)}\{o_1, o_2, a_3\}.
    \end{align*}
    
    \item  If \( y_2^*(a_3) = 0\) for all $a_1$ and $a_2$, then in this case, we have:
     \begin{align*}
        &y_2^*(a_3) o_2+\underset{a_2}{\mathbb{E}}\left[y_1^*(a_2,a_3) o_1 + \underset{a_1}{\mathbb{E}}\left[\sum\limits_{i=1}^{3} x_i^*(\boldsymbol{a})a_i\right]\right] \\
        =& \underset{a_2}{\mathbb{E}}\left[y_1^*(a_2,a_3) o_1 + \underset{a_1}{\mathbb{E}}\left[\sum\limits_{i=1}^{3} x_i^*(\boldsymbol{a})a_i\right]\right] .
    \end{align*}
    We then fix  $a_2$ along with the previously fixed $a_3$ and consider the following two subcases:
     \begin{enumerate}[(2.1)]
         \item If $y_1^*(a_2,a_3)=1$ for all $a_1$, then we have
             \begin{align*}
             &y_1^*(a_2,a_3) o_1 + \underset{a_1}{\mathbb{E}}\left[\sum\limits_{i=1}^{3} x_i^*(\boldsymbol{a})a_i\right] \\
             =&\  o_1 + \max\{a_2,a_3\}\\
             \leq&\  {\max}^{(2)}\left\{o_1,a_2,a_3\right\}.
            \end{align*}
        \item If $y_1^*(a_2,a_3)=0$ for all $a_1$, then we get
             \begin{align*}
             &y_1^*(a_2,a_3) o_1 + \underset{a_1}{\mathbb{E}}\left[\sum\limits_{i=1}^{3} x_i^*(\boldsymbol{a})a_i\right] \\
             &=  \underset{a_1}{\mathbb{E}}\left[\sum\limits_{i=1}^{3} x_i^*(\boldsymbol{a})a_i\right]\\
             &\leq \underset{a_1}{\mathbb{E}}\left[{\max}^{(2)}\{a_1,a_2,a_3\}\right].
            \end{align*}
     \end{enumerate}
    Combining above subcases, we obtain 
    \begin{align*}
        &\underset{a_2}{\mathbb{E}}\left[y_1^*(a_2,a_3) o_1 + \underset{a_1}{\mathbb{E}}\left[\sum\limits_{i=1}^{3} x_i^*(\boldsymbol{a})a_i\right]\right] \leq\\
        &\underset{a_2}{\mathbb{E}}\Big[\max \Big\{{\max}^{(2)}\{o_1,a_2,a_3\},\underset{a_1}{\mathbb{E}}\left[{\max}^{(2)}\{a_1,a_2,a_3\}\right]\Big\} \Big].
    \end{align*}    
\end{enumerate}

To simplify the mathematical expression, we introduce the following notation:
\begin{align}
    \mathrm{m}_{\{1,2\}}^{F} &= {\max}^{(2)}\{o_{1}, o_{2}, a_{3}\}, \notag \\
    \mathrm{m}_{\{1\}}^{F} &= {\max}^{(2)} \{o_{1}, a_{2}, a_{3}\}, \notag \\
    \mathrm{m}_{0} &= {\max}^{(2)} \{a_{1}, a_{2}, a_{3}\}. \label{eqn:notation}
\end{align}

Therefore, by taking the expectation over $a_3$, we obtain an upper bound for the optimal mechanism:
\[
    \operatorname{OBJ}(\mathcal{M}^{*}) \leq 
    \underset{a_3}{\mathbb{E}}\left[
            \max\left\{\mathrm{m}_{\{1,2\}}^F,\underset{a_2}{\mathbb{E}}\left[\max\{\mathrm{m}_{\{1\}}^F,
            \underset{a_1}{\mathbb{E}}\left[\mathrm{m}_0\right]\}\right]\right\}\right].
\]

Using the same notation, we can also derive a lower bound for $\mathcal{M}^F$:
\begin{align*}
        \operatorname{OBJ}(\mathcal{M}^{F}) &\geq \max\left\{ \operatorname{OBJ}(\mathcal{M}_{\{1,2\}}^F),\operatorname{OBJ}(\mathcal{M}_{\{1\}}^F),\operatorname{OBJ}(\mathcal{M}_{0})\right\}\\
        &= \max\left\{ \underset{a_3}{\mathbb{E}}\left[\mathrm{m}_{\{1,2\}}^F\right],\underset{a_2,a_3}{\mathbb{E}}\left[\mathrm{m}_{\{1\}}^F\right],\underset{\boldsymbol{a}}{\mathbb{E}}\left[\mathrm{m}_0\right]\right\}.
\end{align*}

Thus, we can determine the approximate ratio of the G-FIX mechanism relative to the optimal mechanism:
\[
    \frac{\operatorname{OBJ}(\mathcal{M}^F)}{\operatorname{OBJ}(\mathcal{M}^{*})} \geq 
    \frac{\max\left\{ \underset{a_3}{\mathbb{E}}\left[\mathrm{m}_{\{1,2\}}^F\right],\underset{a_2,a_3}{\mathbb{E}}\left[\mathrm{m}_{\{1\}}^F\right],\underset{\boldsymbol{a}}{\mathbb{E}}\left[\mathrm{m}_0\right]\right\}}{ \underset{a_3}{\mathbb{E}}\left[
            \max\left\{\mathrm{m}_{\{1,2\}}^F,\underset{a_2}{\mathbb{E}}\left[\max\{\mathrm{m}_{\{1\}}^F,
            \underset{a_1}{\mathbb{E}}\left[\mathrm{m}_0\right]\}\right]\right\}\right]}.
\]

Before proceeding with the analysis, it is crucial to grasp the expression intuitively: 
Firstly, \(\mathrm{m}_I^F\) represents the outcome of running the corresponding \(\mathcal{M}_I^F\) mechanism with a given set of bids. For example:
\begin{itemize}
    \item \(\mathrm{m}_{\{1,2\}}^F\) represents the outcome of \(\mathcal{M}_{\{1,2\}}^F\) given \(a_3\).
    \item \(\mathrm{m}_{\{1\}}^F\) represents the outcome of \(\mathcal{M}_{\{1\}}^F\) given \(a_2\) and \(a_3\).
    \item \(\mathrm{m}_0\) represents the outcome of \(\mathcal{M}_0\) given \(a_1\), \(a_2\), and \(a_3\).
\end{itemize}

Thus, the numerator can be interpreted as running one of the mechanisms \(\mathcal{M}_{\{1,2\}}^F\), \(\mathcal{M}_{\{1\}}^F\), or \(\mathcal{M}_0\) for all possible \(\boldsymbol{a}\) and choosing one with the best expected outcome. The denominator, on the other hand, can be understood as follows:
\begin{itemize}
    \item Given \(a_3\), the mechanism can decide whether to run \(\mathcal{M}_{\{1,2\}}^F\) or incorporate more information based on the values of \(\mathrm{m}_{\{1,2\}}^F\) and \(\underset{a_2}{\mathbb{E}}\left[\max\left\{\mathrm{m}_{\{1\}}^F, \underset{a_1}{\mathbb{E}}\left[\mathrm{m}_0\right]\right\}\right]\).
    \begin{itemize}
        \item If \(\mathrm{m}_{\{1,2\}}^F\) is larger, \(\mathcal{M}_{\{1,2\}}^F\) is executed.
        \item Otherwise, the mechanism introduces \(a_2\) and, based on the values of \(a_2\) and \(a_3\), decides whether to run \(\mathcal{M}_{\{1\}}^F\) or \(\mathcal{M}_0\).
        \begin{itemize}
            \item If \(\mathrm{m}_{\{1\}}^F\) is larger, \(\mathcal{M}_{\{1\}}^F\) is executed.
            \item If \(\underset{a_1}{\mathbb{E}}\left[\mathrm{m}_0\right]\) is larger, \(\mathcal{M}_0\) is executed.
        \end{itemize}
    \end{itemize}
\end{itemize}
Intuitively, the mechanism in the numerator pre-selects the one with the highest expected outcome across all possible bids, while the mechanism in the denominator is capable of making decisions based on the actual situation step by step.

We will now decompose this fraction based on the step-by-step intuition, breaking it down into a more manageable form.
\begin{align}
    &\frac{\operatorname{OBJ}(\mathcal{M}^F)}{\operatorname{OBJ}(\mathcal{M}^{*})} \geq 
    \notag
    \\
    & \frac{\max \{ 
                          \underset{a_3}{\mathbb{E}}[\mathrm{m}_{\{1,2\}}^F],
                          \underset{a_2,a_3}{\mathbb{E}}[\mathrm{m}_{\{1\}}^F],
                          \underset{\boldsymbol{a}}{\mathbb{E}}[\mathrm{m}_0]\}}{\max\{\underset{a_3}{\mathbb{E}}[\mathrm{m}_{\{1,2\}}^F],
                          \underset{a_3}{\mathbb{E}}[\max\{\underset{a_2}{\mathbb{E}}[\mathrm{m}_{\{1\}}^F],
                          \underset{a_1,a_2}{\mathbb{E}}[\mathrm{m}_0]\}]\}} \label{eqn:frac a}\\
    &\cdot \frac{\max\{  \underset{a_3}{\mathbb{E}}[\mathrm{m}_{\{1,2\}}^F],
                          \underset{a_3}{\mathbb{E}}[\max\{\underset{a_2}{\mathbb{E}}[\mathrm{m}_{\{1\}}^F],
                          \underset{a_1,a_2}{\mathbb{E}}[\mathrm{m}_0]\}]\}}{\max\{
                          \underset{a_3}{\mathbb{E}}[\mathrm{m}_{\{1,2\}}^F],
                          \underset{a_3}{\mathbb{E}}[
                            \underset{a_2}{\mathbb{E}}[
                          \max\{\mathrm{m}_{\{1\}}^F,
                          \underset{a_1}{\mathbb{E}}[\mathrm{m}_0] \}  
                            ]
                          ]  \}  } \label{eqn:frac B}\\
    &\cdot \frac{\max \{       
                          \underset{a_3}{\mathbb{E}}[\mathrm{m}_{\{1,2\}}^F],
                          \underset{a_3}{\mathbb{E}}[
                            \underset{a_2}{\mathbb{E}}[
                                 \max\{   \mathrm{m}_{\{1\}}^F,
                          \underset{a_1}{\mathbb{E}}[\mathrm{m}_0] \}  
                            ]
                          ] \}  }{\underset{a_3}{\mathbb{E}}[\max\{       
                          \mathrm{m}_{\{1,2\}}^F,
                          \underset{a_2}{\mathbb{E}}[
                                 \max\{\mathrm{m}_{\{1\}}^F,
                          \underset{a_1}{\mathbb{E}}[\mathrm{m}_0]\}]\}]} .    \label{eqn:frac c}                 
\end{align}

In the three decomposed fractions, each denominator adds some flexibility compared to the numerator. For example, in fraction (\ref{eqn:frac a}), the denominator adds the flexibility of \(a_3\), allowing the choice of running \(\mathcal{M}_{\{1\}}^F\) or \(\mathcal{M}_0\) based on the value of \(a_3\). In fraction (\ref{eqn:frac B}), given \(a_3\), the denominator adds the flexibility of \(a_2\), enabling the choice of running \(\mathcal{M}_{\{1\}}^F\) or \(\mathcal{M}_0\) based on the value of \(a_2\). Fraction (\ref{eqn:frac c}) adds the flexibility of \(a_1\). 

Next, we will separately prove the lower bound for each of these three fractions. We will start by introducing a mathematical result from the previous paper.
\begin{lemma}[\cite{LMZWZYXZD24}]\label{cl:4/5}
If $\frac{h_2(t)}{h_1(t)} \in [r,r+1]$ for all $t$, where $0 < r < 1$, then for any distribution $G_1$,
\[
\frac{\max \{\underset{a_1\sim G_1}{\mathbb{E}}[h_1(a_1)], \underset{a_1\sim G_1}{\mathbb{E}} [h_2(a_1)]\}}{\underset{a_1 \sim G_1}{\mathbb{E}}[\max \{h_1(a_1), h_2(a_1)\}]} \geq \frac{4}{5}.
\]
\end{lemma}
Next, we will prove the following three claims based on this claim.
\begin{claim}\label{claim: part 1}
    With $\mathrm{m}_{\{1,2\}}^F$, $\mathrm{m}_{\{1\}}^F$ and $\mathrm{m}_0$ defined in notation (\ref{eqn:notation}), we have
    \[
        \frac{\max \{ 
                          \underset{a_3}{\mathbb{E}}[\mathrm{m}_{\{1,2\}}^F],
                          \underset{a_2,a_3}{\mathbb{E}}[\mathrm{m}_{\{1\}}^F],
                          \underset{\boldsymbol{a}}{\mathbb{E}}[\mathrm{m}_0]\}}{\max\{\underset{a_3}{\mathbb{E}}[\mathrm{m}_{\{1,2\}}^F],
                          \underset{a_3}{\mathbb{E}}[\max\{\underset{a_2}{\mathbb{E}}[\mathrm{m}_{\{1\}}^F],
                          \underset{a_1,a_2}{\mathbb{E}}[\mathrm{m}_0]  \}  ]  \}  } \label{eqn:frac A} \geq \frac{4}{5}.
    \]
\end{claim}
\begin{proof}
    Since 
    \[
        \max\{\underset{a_2,a_3}{\mathbb{E}}[\mathrm{m}_{\{1\}}^F],
                          \underset{\boldsymbol{a}}{\mathbb{E}}[\mathrm{m}_0]\} \leq 
                          \underset{a_3}{\mathbb{E}}[\max\{\underset{a_2} {\mathbb{E}}[\mathrm{m}_{\{1\}}^F],
                          \underset{a_1,a_2}{\mathbb{E}}\left[\mathrm{m}_0]  \right\}],
    \]
    deleting $\underset{a_3}{\mathbb{E}}[\mathrm{m}_{\{1,2\}}^F]$ on numerator and denominator, we get 
    \begin{align}
        &\frac{\max \{ 
                          \underset{a_3}{\mathbb{E}}[\mathrm{m}_{\{1,2\}}^F],
                          \underset{a_2,a_3}{\mathbb{E}}[\mathrm{m}_{\{1\}}^F],
                          \underset{\boldsymbol{a}}{\mathbb{E}}[\mathrm{m}_0]\}}{\max\{\underset{a_3}{\mathbb{E}}[\mathrm{m}_{\{1,2\}}^F],
                          \underset{a_3}{\mathbb{E}}[\max\{\underset{a_2}{\mathbb{E}}[\mathrm{m}_{\{1\}}^F],
                          \underset{a_1,a_2}{\mathbb{E}}[\mathrm{m}_0]  \}]\}  } \notag\\
        &\geq \frac{      \max\{
                          \underset{a_2,a_3}{\mathbb{E}}[\mathrm{m}_{\{1\}}^F],
                          \underset{\boldsymbol{a}}{\mathbb{E}}[\mathrm{m}_0]\}}{
                          \underset{a_3}{\mathbb{E}}[\max\{\underset{a_2} {\mathbb{E}}[\mathrm{m}_{\{1\}}^F],
                          \underset{a_1,a_2}{\mathbb{E}}\left[\mathrm{m}_0]  \right\}] }.\label{frac:1}
    \end{align}
    Define 
    \begin{align*}
        h_{1}(t) = \underset{a_{2}}{\mathbb{E}}\left[{\max}^{(2)}\{o_{1}, a_{2}, t\}\right]
    \end{align*}
    and 
    \begin{align*}
        h_{2}(t) = \underset{a_{1}a_{2}}{\mathbb{E}}\left[{\max}^{(2)} \left\{  a_{1}, a_{2}, t\right\} \right],
    \end{align*}
     we can rewrite right hand side of inequality (\ref{frac:1}) as follow:
     \begin{equation*}
         \frac{\max\{ \underset{a_{3}}{\mathbb{E}}[h_{1}(a_{3})],\underset{a_{3}}{\mathbb{E}}[h_{2}(a_{3})]
         \}}{\underset{a_{3}}{\mathbb{E}}[\max \{h_{1}(a_{3}), h_{2}(a_{3})\}]}.
     \end{equation*}
     In order to analyze the value of this fraction, we will calculate the derivatives of \(h_1(t)\) and \(h_2(t)\) and their values at critical points.
     First, for \(h_1(t)\), we can obtain:
     \begin{enumerate}
         \item The derivative of $h_1(t)$:
         \[
                h_{1}^{\prime}(t)=\begin{cases}\mathbb{P}(t \geq a_2) & t<o_{1} \\
                1 & t> o_{1},\end{cases}
         \]
         \item $h_1(t)$\ 's values at critical points:
         \begin{align*}
             h_{1}(0) &=o_{1}+\underset{a_{2}}{\mathbb{E}}\left[a_{2}\right],\\
             h_{1}(o_{1}) &=2o_{1}+\underset{a_{2}}{\mathbb{E}}\left[\mathbb{I}[a_{2}>o_{1}](a_{2}-o_{1})\right],\\
             \lim_{t \to \infty}h_{1}(t)&=\lim_{t \to \infty}t+o_{1}+\underset{a_{2}}{\mathbb{E}}\left[\mathbb{I}[a_{2}>o_{1}](a_{2}-o_{1})\right].
         \end{align*}
     \end{enumerate}
    Similarly for \(h_2(t)\), we can obtain:
    \begin{enumerate}
        \item The derivative of $h_2(t)$:
        \[
             h_{2}^{\prime}(t)=1-\mathbb{P}(t < a_{1})\mathbb{P}(t < a_{2}),
        \]
        \item $h_2(t)$\ 's values at critical points:
        \begin{align*}
             h_{2}(0) &=\underset{a_{1}}{\mathbb{E}}[a_{1}]+\underset{a_{2}}{\mathbb{E}}[a_{2}],\\
             h_{2}(o_{1}) &=a_{1}+a_{2}\\
             &  +\underset{a_{1}a_{2}}{\mathbb{E}}[\mathbb{I}[o_{1}>\min \{ a_{1},a_{2}\}](o_{1}-\min \{ a_{1},a_{2}\})],\\
             \lim_{t \to \infty}h_{2}(t)&=\lim_{t \to \infty}t+\underset{a_{1}a_{2}}{\mathbb{E}}[\max \{a_{1}, a_{2}\}] \\
            &\geq \lim_{t \to \infty}t+\underset{a_{2}}{\mathbb{E}}[\max \{\underset{a_{1}}{\mathbb{E}}[a_{1}], a_{2}\}].      
        \end{align*}
    \end{enumerate}

    From the characterization of the functions above, we can observe that \(h_1(t)\) grows more slowly than \(h_2(t)\) in the interval \((0, o_1)\), and grows faster than \(h_2(t)\) in the interval \((o_1, +\infty)\). Additionally, the relationship between \(h_1(0)\) and \(h_2(0)\), as well as \(h_1(+\infty)\) and \(h_2(+\infty)\), depend on the relationship between \(o_1\) and \({\mathbb{E}_{a_1}}[a_1]\). 
    
    It is evident from Figure \ref{fig:claim_1 (1)} that when \(o_1 \leq \underset{a_1}{\mathbb{E}}[a_1]\), \(h_1(t)\) is always less than or equal to \(h_2(t)\), and when \(o_1 > \underset{a_1}{\mathbb{E}}[a_1]\) and $h_1(o_1)\geq h_2(o_1)$, as shown in Figure \ref{fig:claim_1 (2)} , \(h_1(t)\) is always greater than or equal to \(h_2(t)\). In both cases, we have:
    \[
        \frac{\max \{ \underset{a_{3}}{\mathbb{E}}[h_{1}(a_{3})],\underset{a_{3}}{\mathbb{E}}[h_{2}(a_{3})] 
         \}}{\underset{a_{3}}{\mathbb{E}}[\max \{h_{1}(a_{3}), h_{2}(a_{3})\}]}=1.
    \]
    \begin{figure}[h]
    \centering
        \centering
        \includegraphics[width=0.7\linewidth]{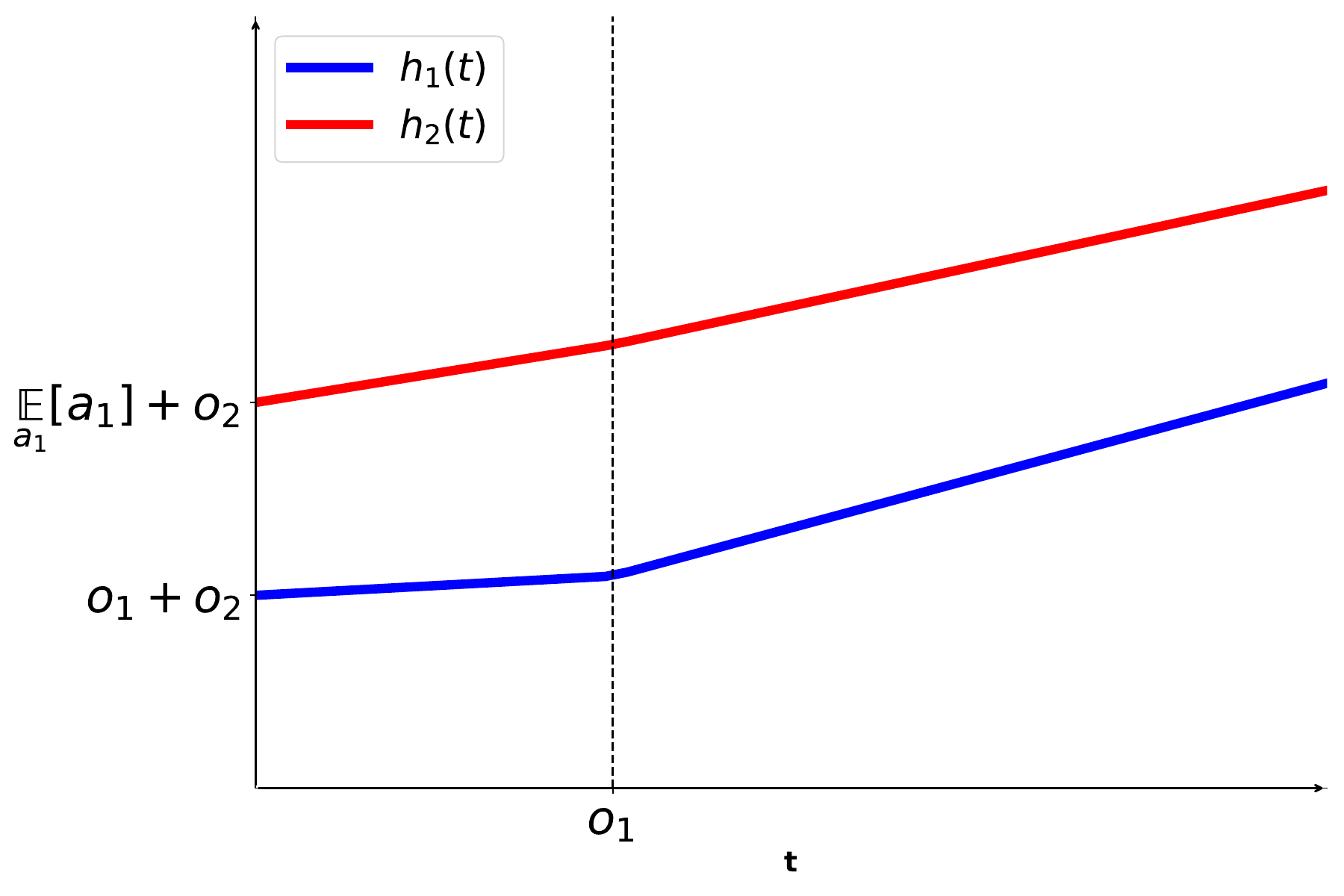}
        \caption{The relationship between $h_1(t)$ and $h_2(t)$ when \(o_1 \leq \underset{a_1}{\mathbb{E}}[a_1]\).}
        \label{fig:claim_1 (1)}
    \end{figure}
    \begin{figure}[h]
        \centering
        \includegraphics[width=0.7\linewidth]{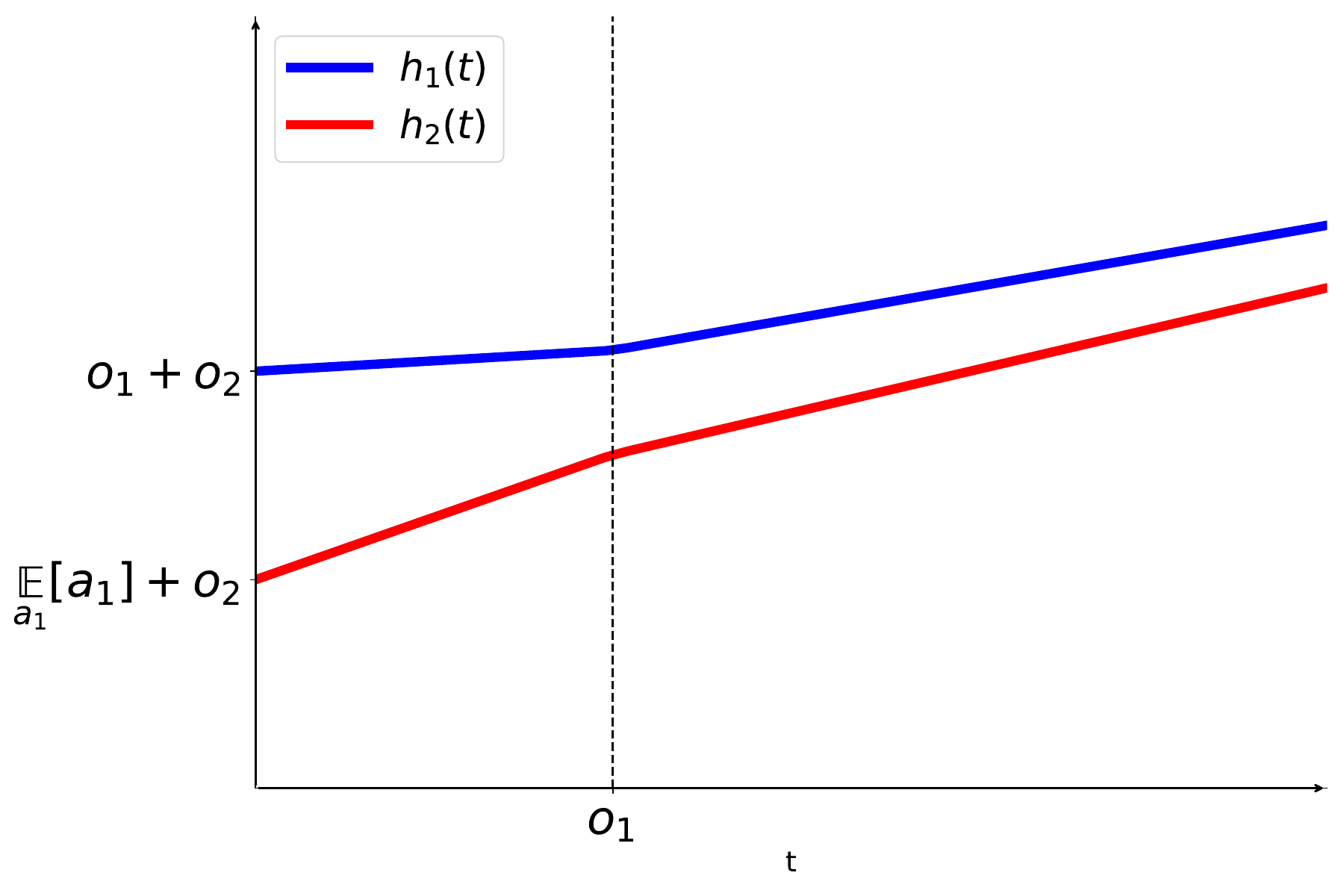}
        \caption{The relationship between $h_1(t)$ and $h_2(t)$ when \(o_1 > \underset{a_1}{\mathbb{E}}[a_1]\) and $h_1(o_1)\geq h_2(o_1)$.}
        \label{fig:claim_1 (2)}
    \end{figure}


    Therefore we only need to consider the case that \(o_1 > {\mathbb{E}_{a_1}}[a_1]\) and $h_1(o_1) < h_2(o_1)$. See Figure \ref{fig:claim_1 (3)} to help us understand. 
    \begin{figure}[h]
        \centering
        \includegraphics[width=0.7\linewidth]{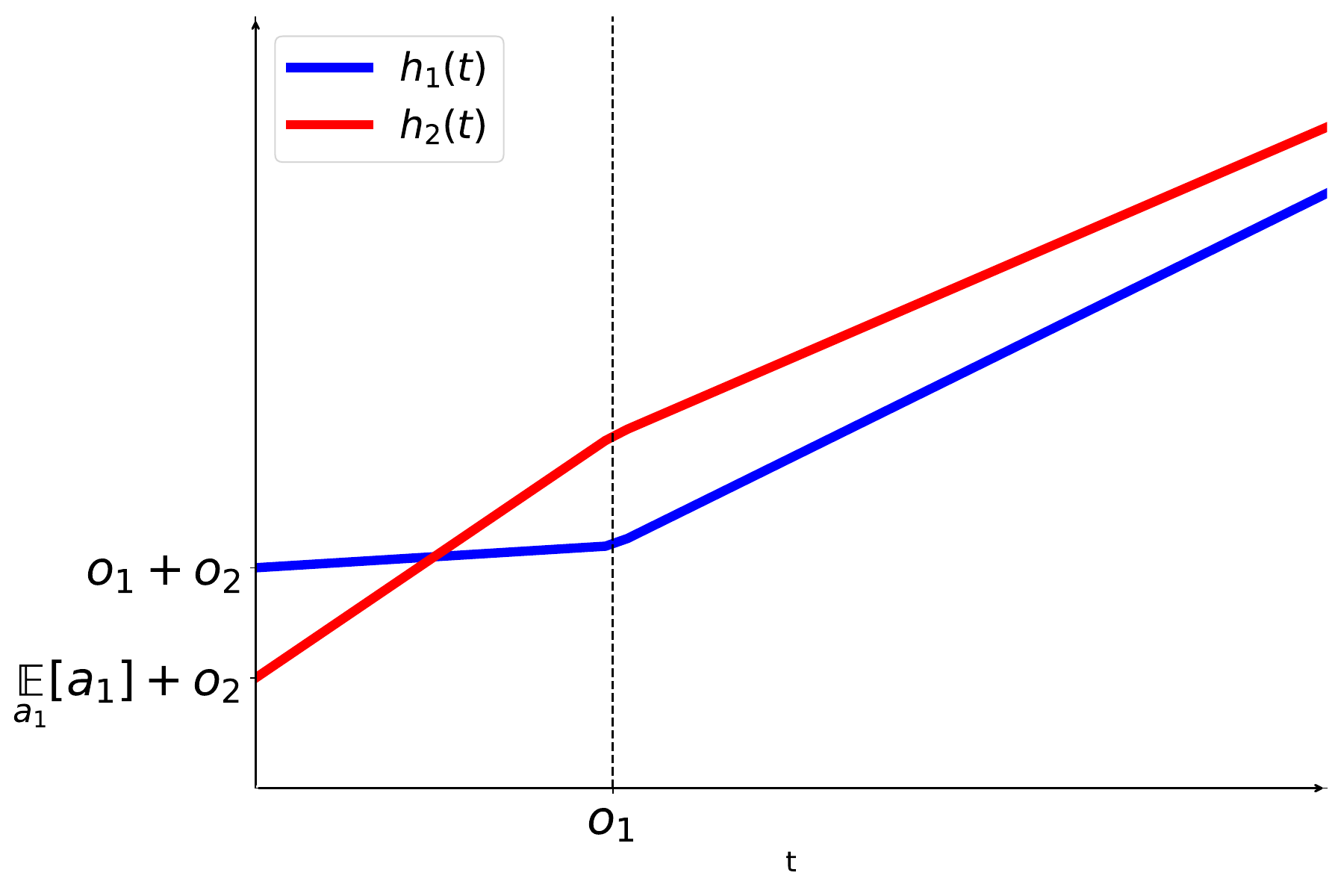}
        \caption{The relationship between $h_1(t)$ and $h_2(t)$ when \(o_1 > \underset{a_1}{\mathbb{E}}[a_1]\) and $h_1(o_1)< h_2(o_1)$.}
        \label{fig:claim_1 (3)}
    \end{figure}
    
    Since \(\frac{h_2(t)}{h_1(t)}\) is monotonically increasing in the interval \((0, o_1)\) and monotonically decreasing in the interval \((o_1, +\infty)\), and $\lim\limits_{t \to \infty}\frac{h_2(t)}{h_1(t)}=1$, let 
    \[
        \frac{h_2(0)}{h_1(0)} = \frac{\underset{a_1}{\mathbb{E}}[a_1]+\underset{a_2}{\mathbb{E}}[a_2]}{o_1+\underset{a_2}{\mathbb{E}}[a_2]}=r \in (0, 1),
    \] 
    and we have:
    \[
        r=\frac{h_2(0)}{h_1(0)}\leq\frac{h_2(t)}{h_1(t)}\leq\frac{h_2(o_1)}{h_1(o_1)}\leq\frac{\underset{a_{1}}{\mathbb{E}}[a_{1}]+\underset{a_{2}}{\mathbb{E}}[a_{2}]+o_{1}}{\underset{a_{2}}{\mathbb{E}}[a_{2}]+o_{1}}\leq r+1.
    \]

    Thus, according to Lemma \ref{cl:4/5}, we have:
    \[
         \frac{\max \{ \underset{a_{3}}{\mathbb{E}}[h_{1}(a_{3})],\underset{a_{3}}{\mathbb{E}}[h_{2}(a_{3})] 
         \}}{\underset{a_{3}}{\mathbb{E}}\left[\max \{h_{1}(a_{3}), h_{2}(a_{3})\}\right]}\geq\frac{4}{5}.
    \]
    By inequality (\ref{frac:1}), we prove this claim.    
\end{proof}

\begin{claim}\label{claim: part 2}
    With $\mathrm{m}_{\{1,2\}}^F$, $\mathrm{m}_{\{1\}}^F$ and $\mathrm{m}_0$ defined in notation (\ref{eqn:notation}), we have
    \[
        \frac{\max\{  \underset{a_3}{\mathbb{E}}[\mathrm{m}_{\{1,2\}}^F],
                          \underset{a_3}{\mathbb{E}}[\max\{\underset{a_2}{\mathbb{E}}[\mathrm{m}_{\{1\}}^F],
                          \underset{a_1,a_2}{\mathbb{E}}[\mathrm{m}_0]\}]\}}{\max\{   
                          \underset{a_3}{\mathbb{E}}[\mathrm{m}_{\{1,2\}}^F],
                          \underset{a_3}{\mathbb{E}}[
                            \underset{a_2}{\mathbb{E}}[
                                 \max\{\mathrm{m}_{\{1\}}^F,
                          \underset{a_1}{\mathbb{E}}[\mathrm{m}_0]\}  
                            ]]\}  } \geq \frac{4}{5}.
    \]
\end{claim}
\begin{proof}
    Since  
    $$\underset{a_3}{\mathbb{E}}[\max\{\underset{a_2}{\mathbb{E}}[\mathrm{m}_{\{1\}}^F], \underset{a_1,a_2}{\mathbb{E}}[\mathrm{m}_0]\}] \leq \underset{a_2,a_3}{\mathbb{E}}[\max\{\mathrm{m}_{\{1\}}^F, \underset{a_1}{\mathbb{E}}\left[\mathrm{m}_0\right] \}],$$
    with the similar argument in proof of Claim \ref{claim: part 1}, we know 
    \begin{align}
        &\frac{\max\{  \underset{a_3}{\mathbb{E}}[\mathrm{m}_{\{1,2\}}^F],
                          \underset{a_3}{\mathbb{E}}[\max\{\underset{a_2}{\mathbb{E}}[\mathrm{m}_{\{1\}}^F],
                          \underset{a_1,a_2}{\mathbb{E}}[\mathrm{m}_0]  \}]\}}{\max\{
                          \underset{a_3}{\mathbb{E}}[\mathrm{m}_{\{1,2\}}^F],
                          \underset{a_3}{\mathbb{E}}[
                            \underset{a_2}{\mathbb{E}}[
                                 \max\{\mathrm{m}_{\{1\}}^F,
                          \underset{a_1}{\mathbb{E}}[\mathrm{m}_0]\}]]\} } \notag\\
        &\geq \frac{ 
                        \underset{a_3}{\mathbb{E}}[\max\{\underset{a_2}{\mathbb{E}}[\mathrm{m}_{\{1\}}^F],
                          \underset{a_1,a_2}{\mathbb{E}}[\mathrm{m}_0]  \}  ]  }{      
                          \underset{a_3}{\mathbb{E}}[
                            \underset{a_2}{\mathbb{E}}[
                                 \max\{\mathrm{m}_{\{1\}}^F,
                          \underset{a_1}{\mathbb{E}}[\mathrm{m}_0]\}]]}.\label{frac:2}
    \end{align}
    Define
    \begin{align*}
        h_{3}(t, a_{3}) = {\max}^{(2)}\{o_{1}, t, a_{3} \}
    \end{align*}
    and
    \begin{align*}
        &h_{4}(t, a_{3}) = \underset{a_{1}}{\mathbb{E}}[{\max}^{(2)} \{a_{1}, t, a_{3} \}],
    \end{align*}
    we can rewrite fraction (\ref{frac:2}) as follow:
    \[
        \frac{ 
                        \underset{a_3}{\mathbb{E}}[\max\{\underset{a_2}{\mathbb{E}}[h_{3}(a_2, a_{3})],
                          \underset{a_2}{\mathbb{E}}[h_{4}(a_2, a_{3})]\}]  }{      
                          \underset{a_3}{\mathbb{E}}[
                            \underset{a_2}{\mathbb{E}}[
                                 \max\{ h_{3}(a_2, a_{3}),
                         h_{4}(a_2, a_{3})\}]] }.
    \]
    To ease our analysis, we will consider two cases: \(a_3 \leq o_1\) and \(a_3 > o_1\). In the first case, for \(h_3(t,a_3)\), we can obtain:
    \begin{enumerate}
        \item The derivative of $h_3(t,a_3)$:
        \[
           h_{3}^{\prime}(t, a_{3})  = 
           \begin{cases}
                0 & t < a_{3} ,\\ 
                1 & t > a_{3} .
            \end{cases}
        \]
        \item $h_3(t,a_3)$\ 's values at critical points:
        \begin{align*}
            h_{3}(0, a_{3}) &= o_{1} + a_{3},\\
            h_{3}(a_{3}, a_{3}) &= o_{1} + a_{3}.
        \end{align*}
    \end{enumerate}
    For \(h_4(t,a_3)\), we can obtain:
    \begin{enumerate}
        \item The derivative of $h_4(t, a_3)$:
        \[
           h_{4}^{\prime}(t, a_{3}) = \begin{cases}
            \mathbb{P}(t \geq a_1) & t< a_{3}, \\ 
            1 & t> a_{3} .
            \end{cases}
        \]
        \item $h_4(t,a_3)$\ 's values at critical points:
        \begin{align*}
            h_{4}(0, a_{3}) &= \underset{a_{1}}{\mathbb{E}}[a_{1}] + a_3 , \\
            h_{4}(a_{3}, a_{3}) &= 2a_{3} + \underset{a_{1}}{\mathbb{E}}\left[\mathbb{I}[a_{1}\geq a_{3}](a_{1}-a_{3})\right].
        \end{align*}
    \end{enumerate}
    We observe that \(h_4(t, a_3)\) grows faster than \(h_3(t, a_3)\) in the interval \((0, a_3)\) (since \(h_3(t, a_3)\) does not grow in this interval), and both grow at the same rate in the interval \((a_3, +\infty)\). Therefore, when \(h_4(0, a_3) \geq h_3(0, a_3)\), it follows that \(h_4(t, a_3) \geq h_3(t, a_3)\) for all \(t\). Besides, when \(h_4(0, a_3) < h_3(0, a_3)\) and \(h_4(a_3, a_3) \leq h_3(a_3, a_3)\), it follows that \(h_4(t, a_3) \leq h_3(t, a_3)\) for all \(t\geq 0\). In both cases, we have:
    \[
        \frac{  \max\{\underset{a_2}{\mathbb{E}}[h_{3}(a_2, a_{3})], \underset{a_2}{\mathbb{E}}[h_{4}(a_2, a_{3})]\}}{ \underset{a_2}{\mathbb{E}}[\max\{ h_{3}(a_2, a_{3}),h_{4}(a_2, a_{3})\}]} = 1.
    \]
    Therefore, we only need to consider the case that \(h_4(0, a_3) < h_3(0, a_3)\) and \(h_4(a_3, a_3) > h_3(a_3, a_3)\). And since \(\frac{h_4(t,a_3)}{h_3(t,a_3)}\) is monotonically increasing in the interval \((0, a_3)\) and monotonically decreasing in the interval \((a_3, +\infty)\), and $\lim\limits_{t \to \infty}[h_4(t)/h_3(t)]=1$, let \(\frac{h_4(0,a_3)}{h_3(0,a_3)} = \frac{\mathbb{E}_{a_1}[a_1]+a_3}{o_1+a_3}=r \in (0, 1)\), we have:
   \begin{align*}
        r&=\frac{h_4(0,a_3)}{h_3(0,a_3)}\leq\frac{h_4(t,a_3)}{h_3(t,a_3)}\leq\frac{h_4(a_3,a_3)}{h_3(a_3,a_3)}\\
        &\leq\frac{\underset{a_{1}}{\mathbb{E}}[a_{1}]+2a_3}{o_{1}+a_3}\leq r+1.
    \end{align*}
    
    Thus, according to claim \ref{cl:4/5}, we have:
    \[
        \frac{\max\{\underset{a_2}{\mathbb{E}}[h_{3}(a_2, a_{3})], \underset{a_2}{\mathbb{E}}[h_{4}(a_2, a_{3})]\}}{ \underset{a_2}{\mathbb{E}}[\max\{ h_{3}(a_2, a_{3}),h_{4}(a_2, a_{3}) \}]}\geq\frac{4}{5}.
    \]

    Next, let us consider the case when \(a_3 > o_1\).
    Similarly, we can obtain:
    \begin{enumerate}
        \item The derivative of $h_3(t,a_3)$
        \[
           h_{3}^{\prime}(t, a_{3})  = 
           \begin{cases}
                0 & t < o_1 ,\\ 
                1 & t> o_1 .
            \end{cases}
        \]
        \item $h_3(t,a_3)$\ 's values at critical points
        \begin{align*}
            h_{3}(0, a_{3}) &= o_{1} + a_{3},\\
            h_{3}(o_1, a_{3}) &= o_{1} + a_{3},\\
            h_{3}(a_3, a_{3}) &= 2a_{3}.
        \end{align*}
    \end{enumerate}
     For \(h_4(t,a_3)\), we can obtain:
    \begin{enumerate}
        \item the derivative of $h_4(t,a_3)$
        \[
           h_{4}^{\prime}(t, a_{3}) = 
           \begin{cases}
                \mathbb{P}(t\geq a_{1}) & t< o_{1} ,\\ 
                \mathbb{P}(t\geq a_{1}) & o_{1}\leq t< a_{3}, \\
                1& t> a_{3}.
            \end{cases}
        \]
        \item $h_4(t,a_3)$\ 's values at critical points
        \begin{align*}
            h_{4}(0, a_{3}) &= \underset{a_{1}}{\mathbb{E}}[a_{1}] + a_3,\\
            h_{4}(o_1, a_{3}) &= o_1+a_{3} + \underset{a_{1}}{\mathbb{E}}\left[\mathbb{I}[a_{1}\geq o_1](a_{1}-o_1)\right],\\
            h_{4}(a_{3}, a_{3}) &= 2a_{3} + \underset{a_{1}}{\mathbb{E}}\left[\mathbb{I}[a_{1}\geq a_{3}](a_{1}-a_{3})\right].
        \end{align*}
    \end{enumerate}

    We observe that \(h_4(t, a_3)\) grows faster than \(h_3(t, a_3)\) in the interval \((0, o_1)\) (since \(h_3(t, a_3)\) has a derivative of 0 in this interval). In the interval \((o_1, a_3)\), \(h_4(t, a_3)\) grows more slowly than \(h_3(t, a_3)\), and in the interval \((a_3, +\infty)\), both functions grow at the same rate. Additionally, we can easily see that \(h_4(o_1, a_3) \geq h_3(o_1, a_3)\) and \(h_4(a_3, a_3) \geq h_3(a_3, a_3)\). Therefore, in the interval \((o_1, +\infty)\), we always have \(h_4(t, a_3) \geq h_3(t, a_3)\). Consequently, when \(h_4(0, a_3) \geq h_3(0, a_3)\), we have \(h_4(t, a_3) \geq h_3(t, a_3)\) for all \(t \in (0, +\infty)\). That means,in this case, we have:
    \[
        \frac{ \max\{\underset{a_2}{\mathbb{E}}[h_{3}(a_2, a_{3})], \underset{a_2}{\mathbb{E}}[h_{4}(a_2, a_{3})]\}}{ \underset{a_2}{\mathbb{E}}[\max\{ h_{3}(a_2, a_{3}),h_{4}(a_2, a_{3}) \}]} = 1.
    \]
    
    Therefore, we only need to consider the case when \(h_4(0, a_3) < h_3(0, a_3)\). In this situation, as can be seen from Figure \ref{fig:claim_2 case_2 (2)}, we can easily observe that \(\frac{h_4(t, a_3)}{h_3(t, a_3)}\) is monotonically increasing in the interval \((0, o_1)\), and monotonically decreasing in the interval \((o_1, +\infty)\). Furthermore, \(\lim\limits_{t \to +\infty} \frac{h_4(t, a_3)}{h_3(t, a_3)} = 1\). So, let \(\frac{h_4(0,a_3)}{h_3(0,a_3)} =r \in (0, 1)\), we have:
    \begin{align*}
        r&=\frac{h_4(0,a_3)}{h_3(0,a_3)}\leq\frac{h_4(t,a_3)}{h_3(t,a_3)}\leq\frac{h_4(o_1,a_3)}{h_3(o_1,a_3)}\\
        &\leq\frac{\underset{a_{1}}{\mathbb{E}}[a_{1}]+a_3+o_1}{o_{1}+a_3}\leq r+1.
    \end{align*}
    \begin{figure}[h]
        \centering
        \includegraphics[width=0.7\linewidth]{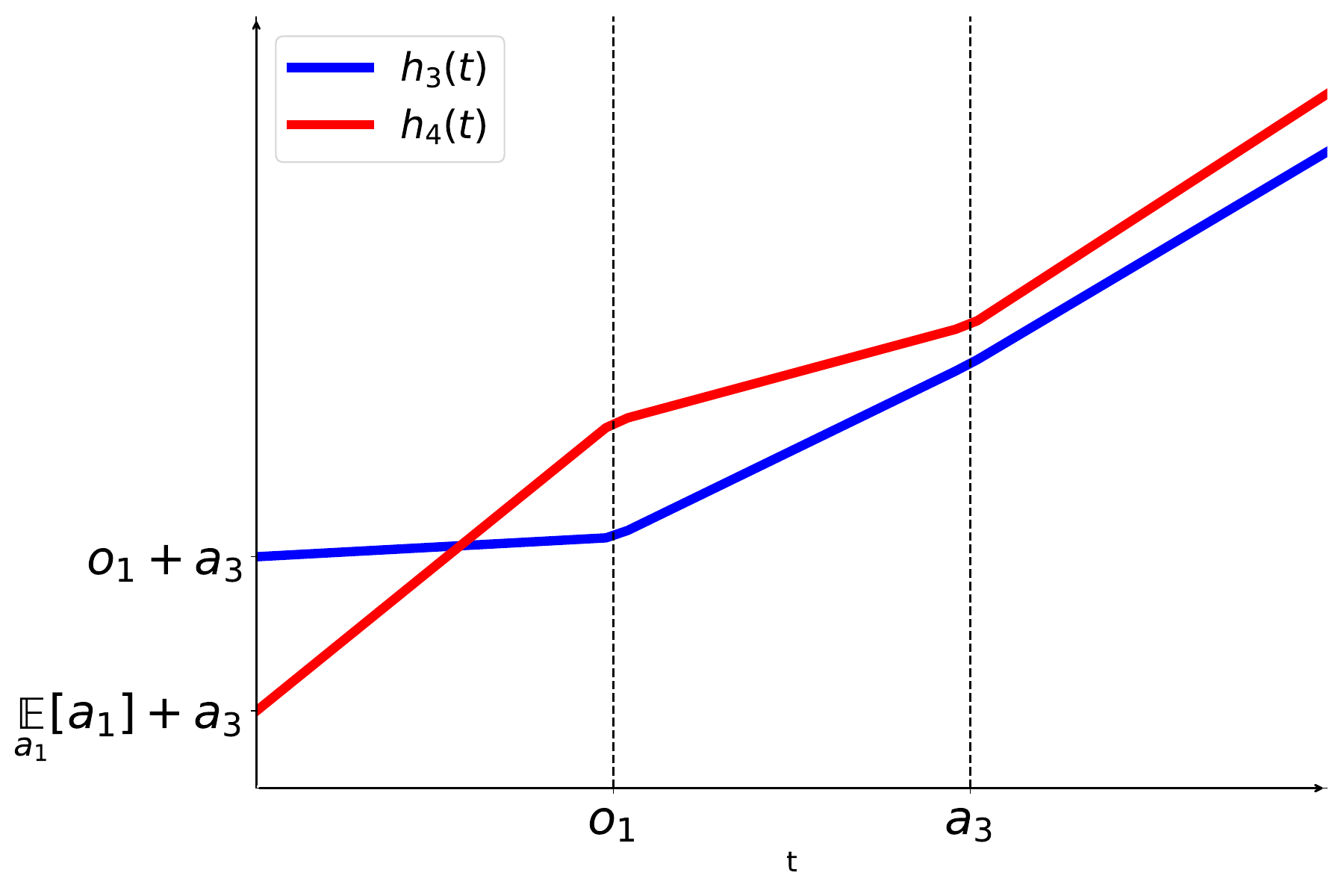}
        \caption{The relationship between $h_3(t)$ and $h_4(t)$ when $a_3>o_1$ and $h_4(0,a_3)< h_3(0,a_3)$.}
        \label{fig:claim_2 case_2 (2)}
    \end{figure}

    Thus, according to claim \ref{cl:4/5}, we have:
    \[
        \frac{\max\{\underset{a_2}{\mathbb{E}}[h_{3}(a_2, a_{3})],\underset{a_2}{\mathbb{E}}[h_{4}(a_2, a_{3})]\}}{ \underset{a_2}{\mathbb{E}}[\max\{ h_{3}(a_2, a_{3}),h_{4}(a_2, a_{3}) \} ]}\geq\frac{4}{5}.
    \]
    Finally, taking expectation of $a_3$ we can get:
    \[
         \frac{ 
                        \underset{a_3}{\mathbb{E}}[\max\{\underset{a_2}{\mathbb{E}}[h_{3}(a_2, a_{3})],
                          \underset{a_2}{\mathbb{E}}[h_{4}(a_2, a_{3})]\}]  }{      
                          \underset{a_3}{\mathbb{E}}[
                            \underset{a_2}{\mathbb{E}}[
                                 \max\{ h_{3}(a_2, a_{3}),
                         h_{4}(a_2, a_{3})\}  
                            ]]}\geq \frac{4}{5}.
    \]
    By inequality (\ref{frac:2}), we show this claim.
\end{proof}

\begin{claim}\label{claim: part 3}
    With $\mathrm{m}_{\{1,2\}}^F$, $\mathrm{m}_{\{1\}}^F$ and $\mathrm{m}_0$ defined in notation (\ref{eqn:notation}), we have
    \[
        \frac{\max\{       
                          \underset{a_3}{\mathbb{E}}[\mathrm{m}_{\{1,2\}}^F],
                          \underset{a_3}{\mathbb{E}}[
                            \underset{a_2}{\mathbb{E}}[
                                 \max\{\mathrm{m}_{\{1\}}^F,
                          \underset{a_1}{\mathbb{E}}[\mathrm{m}_0]\}  
                            ]
                          ] \}}{\underset{a_3}{\mathbb{E}}[\max \{       
                          \mathrm{m}_{\{1,2\}}^F,
                          \underset{a_2}{\mathbb{E}}[\max\{\mathrm{m}_{\{1\}}^F,
                          \underset{a_1}{\mathbb{E}}[\mathrm{m}_0]\}]\} ] } \geq \frac{4}{5}.
    \]

\end{claim}
\begin{proof}
    Define 
    \begin{align*}
        &h_{5}(t) = {\max}^{(2)} \{o_{1}, o_{2}, t \},\\
        &h_{6}(t) = \underset{a_{2}}{\mathbb{E}}[\max \{{\max}^{(2)} \{o_1, a_2, t \},\underset{a_{1}}{\mathbb{E}}[{\max}^{(2)} \{a_{1}, a_{2}, t \}] \}].
    \end{align*}
     We can rewrite fraction in the claim as follow:
     \begin{equation}\label{frac: h style in third}
     \frac{\max \{\underset{a_3}{\mathbb{E}}[h_{5}(a_3)],\underset{a_3}{\mathbb{E}}[h_{6}(a_3)]\}  }{\underset{a_3}{\mathbb{E}}[\max \{  h_{5}(a_3),h_{6}(a_3)\}] }.
     \end{equation}
    To facilitate the analysis, we will consider two cases: \(o_1 \leq \underset{a_1}{\mathbb{E}}\left[a_1\right]\) and \(o_1 > \underset{a_1}{\mathbb{E}}\left[a_1\right]\).
    
    Firstly, we analyze the simpler case where \(o_1 \leq \underset{a_1}{\mathbb{E}}[a_1]\). In this situation, we have:
    \begin{align*}
        h_{6}(t) &= \underset{a_{2}}{\mathbb{E}}\left[\max\left\{{\max}^{(2)} \left\{o_1, a_2, t \right\},\underset{a_{1}}{\mathbb{E}}\left[{\max}^{(2)} \left\{a_{1}, a_{2}, t \right\}\right] \right\}\right]\\
        &=\underset{a_1,a_2}{\mathbb{E}}\left[{\max}^{(2)} \{a_{1}, a_{2}, t \}\right].
    \end{align*}
    For $h_5(t)$, we have:
    \begin{enumerate}
        \item The derivative of $h_5(t)$:
            \[
                h_{5}^{\prime}(t)= \begin{cases}
                    0 & t< \min\{o_1,o_2\},\\
                    1 & t> \min\{o_1,o_2\}.
                \end{cases}
            \]
        \item $h_5(t)$\ 's values at critical points:
        \begin{align*}
            h_5(0) &= o_1+o_2,\\
            h_5(\min\{o_1,o_2\}) &= o_1+o_2,\\
            \lim_{t \to \infty}h_5(t) &= \lim_{t \to \infty}t+\max\{o_1,o_2\}.
        \end{align*}
    \end{enumerate}
    For $h_6(t)$, we have:
    \begin{enumerate}
        \item The derivative of $h_6(t)$
        \[
            h_{6}'(t) = 1-\mathbb{P}(t<a_1)\mathbb{P}(t<a_2).
        \]
        \item $h_6(t)$\ 's values at critical points
        \begin{align*}
            h_6(0) &= \underset{a_1}{\mathbb{E}}[a_1] + \underset{a_2}{\mathbb{E}}[a_2],\\
            \lim_{t \to \infty}h_6(t) &= \lim_{t \to \infty}t+\underset{a_1,a_2}{\mathbb{E}}[\max\{a_1,a_2\}].
        \end{align*}
    \end{enumerate}
    
    We observe that in the interval \((0, \min\{o_1, o_2\})\), \(h_6(t)\) grows faster than \(h_5(t)\) (since \(h_5(t)\) has a growth rate of 0 in this interval). In the interval \((\min\{o_1, o_2\}, +\infty)\), \(h_6(t)\) grows more slowly than \(h_5(t)\), and \(\lim\limits_{t \to \infty} \frac{h_6(t)}{h_5(t)} = 1\). It is straightforward to see that if \(h_6(0) \geq h_5(0)\), then \(h_6(t) \geq h_5(t)\) holds for all \(t\). Conversely, if \(h_6(0) < h_5(0)\) and \(h_6(\min\{o_1, o_2\}) \leq h_5(\min\{o_1, o_2\})\), then \(h_6(t) \leq h_5(t)\) holds for all \(t\). Therefore, we only need to consider the case where \(h_6(0) < h_5(0)\) and \(h_6(\min\{o_1, o_2\}) > h_5(\min\{o_1, o_2\})\). In this case, we can see that \(\frac{h_6(t)}{h_5(t)}\) is increasing in \((0, \min\{o_1, o_2\})\) and decreasing in \((\min\{o_1, o_2\}, +\infty)\). Thus, denote $\frac{h_6(0)}{h_5(0)}=r \in (0,1)$ we have:
    \begin{align*}
         r&=\frac{h_6(0)}{h_5(0)}\leq\frac{h_6(t)}{h_5(t)}\leq\frac{h_6(min\{o_1,o_2\})}{h_5(\min\{o_1,o_2\})}\\
         &\leq\frac{\underset{a_{1}}{\mathbb{E}}[a_{1}]+\underset{a_{1}}{\mathbb{E}}[a_{2}]+\min\{o_1,o_2\}}{o_1+o_2}\leq r+1.
    \end{align*}

   Thus, according to claim \ref{cl:4/5}, we have fraction (\ref{frac: h style in third}) is at least $\frac{4}{5}$ when $o_1\leq \underset{a_1}{\mathbb{E}}[a_1]$.

   Next, we analyze the case when \(o_1 > \underset{a_1}{\mathbb{E}}[a_1]\).
   Denote
   \begin{align*}
       &A =  {\max}^{(2)} \{o_{1}, a_{2}, t \},\\
       &B = \underset{a_{1}}{\mathbb{E}}[\{{\max}^{(2)} \{a_{1}, a_{2}, t \}].
   \end{align*}
   For $h_6(t)$, we have
   \begin{enumerate}
       \item The derivative of $h_6(t)$:\\
       when $t<o_1$:
       \[
            \mathbb{I}[A \geq B] \mathbb{P}(t>a_2)+\mathbb{I}[A < B] \left(1 - \mathbb{P}(t < a_1) \mathbb{P}(t<a_2)\right),
       \]
       
        when $t>o_1$:
       \[
            \mathbb{I}[A \geq B]+\mathbb{I}[A < B] \left(1 - \mathbb{P}(t < a_1) \mathbb{P}(t<a_2)\right).
       \]
       \item $h_6(t)$\ 's values at critical points:
       \begin{align*}
           &h_6(0) = o_1+\underset{a_2}{\mathbb{E}}[a_2].
       \end{align*}
   \end{enumerate}
   
   From the above analysis, we can observe that the derivative of \(h_6(t)\) always lies within the range \((0, 1]\). Therefore, in the interval \((0, \min\{o_1, o_2\})\), \(h_6(t)\) grows faster than \(h_5(t)\) (since \(h_5'(t) = 0\)), while in the interval \((\min\{o_1, o_2\}, +\infty)\), \(h_6(t)\) always grows slower than \(h_5(t)\), and \(\lim\limits_{t \to +\infty} \frac{h_6(t)}{h_5(t)} = 1\). Consequently, it is straightforward to deduce that when \(h_6(0) \geq h_5(0)\), we always have \(h_6(t) \geq h_5(t)\); and when \(h_6(0) < h_5(0)\) and \(h_6(\min\{o_1, o_2\}) \leq h_5(\min\{o_1, o_2\})\), we always have \(h_6(t) \leq h_5(t)\). Therefore, the only case we need to consider is when \(h_6(0) < h_5(0)\) and \(h_6(\min\{o_1, o_2\}) > h_5(\min\{o_1, o_2\})\). In this scenario, as shown in Figure \ref{fig:claim_3 case_2 (3)}, \(\frac{h_6(t)}{h_5(t)}\) is monotonically increasing in the interval \((0, \min\{o_1, o_2\})\) and monotonically decreasing in the interval \((\min\{o_1, o_2\}, +\infty)\). Let \(\frac{h_6(0)}{h_5(0)} = r \in (0,1)\), we have:
    \begin{align*}
         r&=\frac{h_6(0)}{h_5(0)}\leq\frac{h_6(t)}{h_5(t)}\leq\frac{h_6(\min\{o_1,o_2\})}{h_5(\min\{o_1,o_2\})}\\
         &\leq\frac{\underset{a_{1}}{\mathbb{E}}[a_{1}]+\underset{a_{1}}{\mathbb{E}}[a_{2}]+\min\{o_1,o_2\}}{o_1+o_2}\leq r+1.
    \end{align*}
   Thus, according to claim \ref{cl:4/5}, we have fraction (\ref{frac: h style in third}) is at least $\frac{4}{5}$ when $o_1> \underset{a_1}{\mathbb{E}}[a_1]$.
    \begin{figure}[h]
        \centering
        \includegraphics[width=0.7\linewidth]{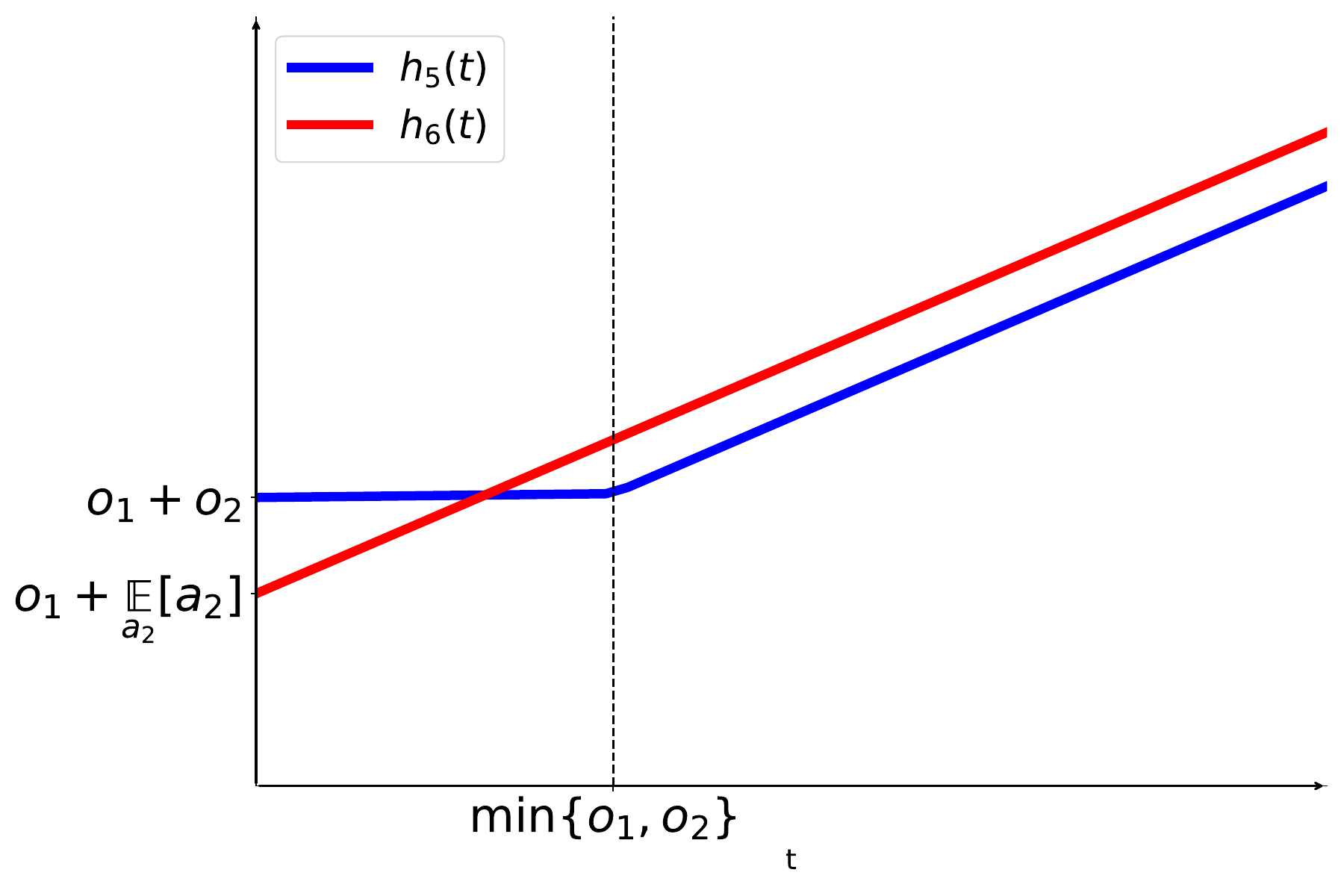}
        \caption{The relationship between $h_5(t)$ and $h_6(t)$ when $o_1>\underset{a_1}{\mathbb{E}}[a_1]$, $h_6(0)< h_5(0)$ and $h_6(\min\{o_1,o_2\})> h_5(\min\{o_1,o_2\})$.}
        \label{fig:claim_3 case_2 (3)}
    \end{figure}
   
   Therefore, we have:
   \[
        \frac{\max\{       
                          \underset{a_3}{\mathbb{E}}[\mathrm{m}_{\{1,2\}}^F],
                          \underset{a_3}{\mathbb{E}}[
                            \underset{a_2}{\mathbb{E}}[
                                 \max\{ \mathrm{m}_{\{1\}}^F,
                          \underset{a_1}{\mathbb{E}}[\mathrm{m}_0]\}  
                            ]]\}}{\underset{a_3}{\mathbb{E}}[\max \{       
                          \mathrm{m}_{\{1,2\}}^F,
                          \underset{a_2}{\mathbb{E}}[
                                 \max\{\mathrm{m}_{\{1\}}^F,
                          \underset{a_1}{\mathbb{E}}\left[\mathrm{m}_0\right]\}]\}]} \geq \frac{4}{5}.
   \]
   Hence, the claim is proved.
\end{proof}
Therefore, combining the conclusions of Claims \ref{claim: part 1}, \ref{claim: part 2}, and \ref{claim: part 3}, we can obtain
\begin{align*}
     & \frac{\operatorname{OBJ}(\mathcal{M}^{F})}{\operatorname{OBJ}(\mathcal{M}^{*})} \\
    \geq &
    \frac{\max\{ \underset{a_3}{\mathbb{E}}[\mathrm{m}_{\{1,2\}}^F],\underset{a_2,a_3}{\mathbb{E}}[\mathrm{m}_{\{1\}}^F],\underset{\boldsymbol{a}}{\mathbb{E}}[\mathrm{m}_0]\}}{ \underset{a_3}{\mathbb{E}}[
            \max\{\mathrm{m}_{\{1,2\}}^F,\underset{a_2}{\mathbb{E}}[\max\{\mathrm{m}_{\{1\}}^F,
            \underset{a_1}{\mathbb{E}}[\mathrm{m}_0]\}]\}]}  \\
    \geq & (\frac{4}{5})^{3}. \qedhere
\end{align*}
\end{proof}

\subsection{Proof of Lemma \ref{lem:c_(n-k)/c_n}}
\begin{proof}
Firstly, define 
\begin{align*}
f(x)=&C_{n}^{l}\left[1-\sum_{i=1}^{l} C_{n-k}^{l-i} x^{(n-k)-(l-i)}(1-x)^{l-i}\right] \\
-& C_{n-k}^{l}\left[1-\sum_{i=1}^{l} C_{n}^{l-i} x^{n-(l-i)}(1-x)^{l-i}\right].
\end{align*}
We can find that:
\begin{align*}
    f(0)&=C_{n}^{l}-C_{n-k}^{l}\geq 0 ,\\
    f(1)&=0,\\
    f^{\prime}(x)&= -g(x)+h(x),
\end{align*}
where
\begin{align*}
    &g(x)=C_{n}^{l} C_{n-k}^{l-1}[(n-k)-(l-1)] x^{(n-k)-(l-1)-1}(1-x)^{l-1}\\
    \text{and} \\
    &h(x)=C_{n-k}^{l} C_{n}^{l-1}[n-(l-1)] x^{n-(l-1)-1}(1-x)^{l-1}.
\end{align*}
Since
\[
    C_{n}^{l} C_{n-k}^{l-1}\left[(n-k)-(l-1)]=C_{n-k}^{l} C_{n}^{l-1}[n-(l-1)\right],
\]
we have
\begin{align*}
    &\frac{g(x)}{h(x)} = x^{-k}>1,x \in (0,1),\\
    &f^{\prime}(x)= -g(x)+h(x)<0, x \in (0,1).
\end{align*}
Thus, \(f(x)\) is monotonically decreasing on the interval \((0, 1)\), and we have:

\[
\frac{1-\sum\limits_{i=1}^lC_{n-k}^{l-i}x^{(n-k)-(l-i)}(1-x)^{(l-i)}}{1-\sum\limits_{i=1}^lC_{n}^{l-i}x^{n-(l-i)}(1-x)^{(l-i)}} \geq \frac{C_{n-k}^{l}}{C_{n}^{l}}. \qedhere
\]
\end{proof}

\subsection{Full Proof of Theorem \ref{thm: fix k from n}}
\begin{proof}
Without loss of generality, we assume that $o_1 \geq o_2$ $ \geq \cdots \geq o_n$ in this proof. Firstly, observe the objective of merging problem (\ref{fml:optimization 2}). We know that for the optimal mechanism $\mathcal{M}^*$, its objective is always bounded by ${\max}^{(k)}\{o_1\cdots o_k, a_1\cdots a_n\}$, i.e.,
\[
\operatorname{OBJ}(\mathcal{M}^{*})  \leq \underset{\boldsymbol{a}\sim G}{\mathbb{E}}\left[{\max}^{(k)}\{o_{1} \cdots o_{k}, a_{1} \cdots a_{n}\}\right].
\]
Let \(\lambda_i\) be the \(i\)-th largest number in the set $\{o_1, \cdots, o_k,$ $a_1,\cdots, a_n\}$. We can rewrite the above inequality as:
\begin{align}
    \operatorname{OBJ}(\mathcal{M}^{*})  &\leq \underset{\boldsymbol{a}\sim G}{\mathbb{E}}\left[{\max}^{(k)}\{o_{1} \cdots o_{k}, a_{1} \cdots a_{n}\}\right] \nonumber\\
    &=\sum\limits_{i=1}^{k}\int_{0}^{\infty}\underset{\boldsymbol{a}\sim G}{\mathbb{P}}\left( \lambda_i\geq t\right)dt. \label{ieq: fix-opt upper}
\end{align}

Since all $a_i$ are identical, for $i\in [k]$,  we can calculate  $\underset{\boldsymbol{a}\sim G}{\mathbb{P}}\left( \lambda_i\geq t\right)$ as
\begin{align*}
    &\underset{\boldsymbol{a}\sim G}{\mathbb{P}}(\lambda_1\geq t)=\begin{cases}
        1 & 0\leq t \leq o_{1} ,\\
        1-G^{n}(t) & t>o_1. 
    \end{cases}\\
    &\underset{\boldsymbol{a}\sim G}{\mathbb{P}}(\lambda_2\geq t)=\\
    &\begin{cases}
        1 & 0\leq t \leq o_2 ,\\
        1-G^{n}(t) & o_{2}<t \leq o_{1} ,\\
        1-\sum\limits_{i=1}^{2}C_n^{2-i}G^{n-(2-i)}(t)(1-G(t))^{2-i} & t>o_{1}.
    \end{cases}\\
    &\ldots\\
    &\underset{\boldsymbol{a}\sim G}{\mathbb{P}}(\lambda_k\geq t)=\\
    &\begin{cases}
        1 & 0\leq t \leq o_k ,\\
        1-\sum\limits_{i=j}^{k}C_n^{k-i}G^{n-(k-i)}(t)(1-G(t))^{(k-i)} & o_{j}<t\leq o_{j-1}\\
         &j \in [2,\cdots,k],\\
        1-\sum\limits_{i=1}^{k}C_n^{k-i}G^{n-(k-i)}(t)(1-G(t))^{(k-i)} & t>o_1.
    \end{cases}
\end{align*}

Plugging all $\underset{\boldsymbol{a}\sim G}{\mathbb{P}}\left( \lambda_i\geq t\right)$ into inequality (\ref{ieq: fix-opt upper}), we get an upper bound of the optimal objective. 
\begin{align}\label{ieq2: fix-opt upper}
    &\operatorname{OBJ}(\mathcal{M}^{*})  \leq \sum\limits_{i=1}^{k}\int_{0}^{\infty}\underset{\boldsymbol{a}\sim G}{\mathbb{P}}( \lambda_i\geq t)dt\notag\\
    =&\sum\limits_{i=1}^{k}o_i+\int_{o_{k}}^{\infty}(1-G^{n}(t)) d t +\cdots \notag \\
    +& \int_{o_{1}}^{+\infty}[1-\sum_{i=1}^{k} C_{n}^{k-i} G^{n-(k-i)}(t)(1-G(t))^{k-i}] d t .
\end{align}

On the other hand, we want to find a lower bound on the objective of $\mathcal{M}^F$. Since $\mathcal{M}^F$ selects the best mechanism among $\mathcal{M}_I^F$, herein, we focus on a specific $\mathcal{M}_I^F$ with $I=\{1,2,\cdots, k\}$, whose objective is ${\max}^{(k)}\{o_1\cdots o_k,a_{k+1}\cdots a_n\}$. Denote \(\theta_i\) by the \(i\)-th largest number in the set $\{o_1\cdots o_k,a_{k+1}\cdots a_n\}$. We have
\begin{align}
     \operatorname{OBJ}(\mathcal{M}^{F})  &\geq \underset{a\sim G}{\mathbb{E}}\left[{\max}^{(k)}\{o_{1} \cdots o_{k}, a_{k+1} \cdots a_{n}\}\right] \notag\\
     =&\sum\limits_{i=1}^{k}\int_{0}^{\infty}\underset{\boldsymbol{a}\sim G}{\mathbb{P}}( \theta_i\geq t)dt\notag\notag
\end{align}
Note that there are only $n-k$ identical $a_i, i \in \{k+1,...,n\}$. With the similar calculation, we get a lower bound of $\operatorname{OBJ}(\mathcal{M}^{F})$, i.e., 
\begin{align}\label{ieq: fix-fix lower}
    &\operatorname{OBJ}(\mathcal{M}^{F}) \geq \sum\limits_{i=1}^{k}o_i+\int_{o_{k}}^{\infty}(1-G^{n-k}(t)) d t +\cdots \notag\\
    +& \int_{o_{1}}^{+\infty}[1-\sum_{i=1}^{k} C_{n-k}^{k-i} G^{(n-k)-(k-i)}(t)(1-G(t))^{k-i}] d t .
\end{align}
It is not hard to find that inequalities (\ref{ieq2: fix-opt upper}) and (\ref{ieq: fix-fix lower}) have the same number of terms. By Lemma \ref{lem:c_(n-k)/c_n} and the following inequality 
\[
    1 = \frac{C_{n-k}^{0}}{C_{n}^{0}}>\cdots>\frac{C_{n-k}^{l}}{C_{n}^{l}}>\frac{C_{n-k}^{l+1}}{C_{n}^{l+1}}>\cdots>\frac{C_{n-k}^{k}}{C_{n}^{k}},
\]
we can replace the terms in inequality (\ref{ieq: fix-fix lower}) by the corresponding terms in inequality (\ref{ieq2: fix-opt upper}) with ${C_{n-k}^{k}}/{C_{n}^{k}}$ times and we get
\begin{align*}
    &\operatorname{OBJ}(\mathcal{M}^{F})  \geq 
   \frac{C_{n-k}^{k}}{C_{n}^{k}}\bigg[\sum\limits_{i=1}^{k}o_i+\int_{o_{k}}^{\infty}(1-G^{n}(t)) d t +\cdots \\
    +& \int_{o_{1}}^{+\infty}[1-\sum_{i=1}^{k} C_{n}^{k-i} G^{n-(k-i)}(t)(1-G(t))^{k-i}] d t \bigg]\\
    =& \frac{C_{n-k}^{k}}{C_{n}^{k}}\operatorname{OBJ}(\mathcal{M}^{*}).
\end{align*}


In conclusion, we show that $\mathcal{M}^{F}$ is ${C_{n-k}^{k}}/{C_{n}^{k}}$ competitive relative to the optimal mechanism $\mathcal{M}^{*}$.
\end{proof} 
\end{document}